\documentclass[11pt,letterpaper]{article}
\usepackage[latin1]{inputenc}
\usepackage[margin=1.0in]{geometry}
\usepackage{palatino}
\usepackage{macros}
\usepackage{csquotes,fullpage}

\DeclareMathOperator{\PPP}{P}
\DeclareMathOperator{\PNP}{NP}
\DeclareMathOperator{\PRP}{RP}

\def\CC{\mathbb{C}}

\def\F{\mathbb{F}}
\def\per{\textup{per}}

\DeclareMathOperator{\range}{Range}

\DeclareMathOperator{\SDP}{SDP}
\def\N{\mathcal{N}}
\usepackage{authblk}

\title{On approximability of the \\Permanent of PSD matrices}
\author{Farzam Ebrahimnejad}
\affil{\small University of Washington, \textsf{febrahim@cs.washington.edu}}
\author{Ansh Nagda}
\affil{\small UC Berkeley, \textsf{anshnagda@gmail.com}}
\author{Shayan Oveis Gharan}
\affil{\small University of Washington, \textsf{shayan@cs.washington.edu}} 
\begin{document}
\maketitle 
\begin{abstract}
    We study the complexity of approximating the permanent of a positive semidefinite matrix $A\in \C^{n\times n}$.
    \begin{enumerate}
        \item We design a new approximation algorithm for $\per(A)$ with approximation ratio $e^{(0.9999 + \gamma)n}$, exponentially improving upon the current best bound of $e^{(1+\gamma-o(1))n}$ \cite{AGOS17,YP22}. Here, $\gamma \approx 0.577$ is Euler's constant.
        \item We prove that it is $\PNP$-hard to approximate $\per(A)$ within a factor $e^{(\gamma-\eps)n}$ for any $\eps>0$. This is the first exponential hardness of approximation for this problem. Along the way, we prove optimal hardness of approximation results for the $\|\cdot\|_{2\to q}$ ``norm'' problem of a matrix for all $-1 < q < 2$.
    \end{enumerate}
\end{abstract}
\thispagestyle{empty}

\newpage
\clearpage
\setcounter{page}{1}

\section{Introduction}
%Let $f:\R_{\geq 0}\to \R$ be continuous and injective. 
 Given a matrix $A\in \C^{n\times n}$, the permanent of $A$ is defined as
$$ \per(A) = \sum_{\sigma\in \mathbb{S}_n} \prod_{i=1}^n A_{i,\sigma_i},$$
where the sum is over all permutations over $n$ elements.
It is well-known that the permanent of a matrix with non-negative entries can be approximated up to a $1+\eps$-multiplicative factor using the MCMC method \cite{JSV04}. Recently there has been significant interest in studying permanent of Hermitian PSD matrices because of close connections to quantum optics and Boson sampling. 
A folklore algorithm is to simply take the product of the entries of the main diagonal to get an $n!$-approximation.

A few years ago, \cite{AGOS17} obtained the first (deterministic) simply exponential approximation algorithms with approximation factor $e^{(\gamma+1)n}$. 
The algorithm proposed in \cite{AGOS17} uses a basic SDP relaxation for the problem; many experts expected that perhaps by using higher-level SDP relaxations one can improve the approximation factor.
Later on, several groups attempted to improve the approximation factor (see e.g., \cite{Bar20}), but for the general case, only subexponential improvements to the approximation ratio were found \cite{YP21,YP22}.
Very recently,  Meiburg showed that contrarily to the permanent of non-negative matrices, it is $\PNP$-Hard to approximate the permanent of a PSD matrix within a factor of $e^{-n^{1-\eps}}$ for any $\eps>0$ \cite{Mei23}. So, the MCMC method falls short of providing a $1+\eps$-approximation for PSD permanents. 

It remained an open problem if, perhaps by using randomness or higher level SDP relaxations, one can obtain an $e^{-\eps n}$  approximation factor for $\eps$ arbitrarily small, or at the very least whether the $\gamma+1$ factor in the exponent can be improved to a smaller constant. We answer both these questions in our work.

Our first result is an exponential improvement on the $e^{-(\gamma + 1)n}$ approximation algorithm mentioned above.

\begin{theorem}[Main Algorithmic Result]\label{thm:upper}
    There is a deterministic polynomial time $e^{-(\gamma + 0.9999)n}$-approximation algorithm for the permanent of a Hermitian PSD matrix $A\in \C^{n\times n}$.
\end{theorem}
% We emphasize that the constant $0.9999$ in the exponent is not optimized, and it is likely that a more careful analysis of our algorithm will yield a more significant improvement.

Our second result is the first exponential hardness of approximation for this problem. As a corollary of a general hardness of approximation result we prove (see \cref{thm:main-thm} below), we show the following:
\begin{theorem}[Main Hardness Result]\label{thm:per-thm}
For all $\eps>0$, it is $\PNP$-hard to approximate the permanent of a Hermitian PSD matrix $A\in \C^{n\times n}$ within a factor $e^{-(\gamma-\eps)n}$.
\end{theorem}
In particular, the above theorem shows that assuming $\PNP\neq \PRP$ {\em even using randomness} the approximation factor of \cite{AGOS17,YP21} cannot be improved by more than a factor of $e^n$.

% To complement our hardness result we propose a new approach with the hope of exponentially improving upon the current $e^{-(\gamma+1)n}$ approximation factor. Our second result shows that assuming the following conjecture on permanents of PSD matrices there is a deterministic polynomial time $e^{-(\gamma+1-\eps)n}$-approximation algorithm for some $\eps>0$.

% \begin{conjecture}\label{conj:main}
%     Let $\alpha = 0.9999$. There exists a $\delta>0$ such that every $n\times n$ matrix $0\preceq A\preceq I$ with $\|A\|_F^2\leq \alpha n$ satisfies $\per(A)\leq e^{-\delta n}$.
% \end{conjecture}
% \begin{theorem}[Main Upper Bound]\label{thm:upper}
%     Assuming \cref{conj:main}, there is a polynomial time algorithm that given a Hermitian PSD matrix $A$, outputs an $e^{-(\gamma+1-\eps)n}$-approximation to $\per(A)$ for $\eps=\min\{\delta,0.002\}$, where $\delta>0$ is the constant in \cref{conj:main}.
% \end{theorem}
% We emphasize that the constant $0.002$ is not optimized, and can be improved through more careful analysis. 
% Let us give a few remarks about the above conjecture. Since $A\preceq I$, then it follows that $\per(A)\leq \per(I)=1$. So, the above conjecture asks to show that as the eigenvalues of $A$ decrease starting from $1$, $\log\per(A)$ decreases linearly.
% We give partial evidence towards this conjecture by showing that it is true when $\alpha = 0.05$:
% \begin{theorem}\label{thm:partial-conjecture}
%      For any every $n\times n$ Hermitian matrix $0\preceq A\preceq I$ with $\|A\|_F^2\leq 0.05 n$ we have 
%      $$\per(A)\leq e^{-0.01 n}.$$
% \end{theorem}

\paragraph{Maximizing Product of Linear Forms} Our hardness techniques also apply to an optimization problem that happens to be related to the permanent of PSD matrices, called the ``maximizing product of linear forms'' problem, studied by Yuan and Parrilo \cite{YP22,YP21}: Given a matrix $V\in \C^{n\times d}$ with rows $v_1,\dots,v_n\in \C^d$, define
\begin{equation}\label{eqn:max-linear-form}
   r(V):=\max_{x\in \C^n:\norm{x}_2=1} \prod_{i=1}^n |\langle x, v_i\rangle|^2.
\end{equation}
They design a polynomial time $O(e^{-\gamma n\cdot (1-o(1))})$-approximation algorithm for $r(V)$ using semidefinite programming, where $\gamma\approx 0.577$ is the Euler-Mascheroni constant. They also prove APX-hardness for this problem, and  raise an open problem of finding the true approximability of this problem. Recently, Meiberg studied an equivalent problem under the name Approximate Quantum Maximum Likelihood Estimation, and showed NP-hardness of approximating it to within any constant factor \cite{Mei23}. Our main technical hardness result, \cref{thm:main-thm}, immediately implies that the maximizing product of linear forms problem (and the Approximate Quantum Maximum Likelihood Estimation problem) does not admit a $e^{-\gamma n(1+\eps)}$-approximation for any constant $\eps>0$, answering the question of Yuan and Parrilo up to sub-exponential factors in $n$.

\subsection{Technical Contributions}\label{sec:technical-contributions}

\subsubsection{Algorithmic results}\label{sec:alg-technical}

In this part, we show the main ideas behind \cref{thm:upper}. We start by presenting algorithms used by previous work. Let $A=VV^\dagger$ be an $n\times n$ PSD matrix where $v_1,\dots,v_n\in \C^n$ are the rows of $V$. Previous work (\cite{AGOS17,YP22}) showed that the value of the following SDP gives a $e^{-(\gamma+1)n}$ approximation to $\per(A)$:
\[\SDP(V) := \max_{X \succeq 0, \tr(X) = n} \prod_{i\in [n]}  v_i^\dag X v_i.\]

$\SDP(V)$ might seem completely unrelated to the definition of $\per(A)$, but we remark that their relationship is a lot more clear when $\per(A)$ is rewritten using Wick's formula (\cref{lem:cnd}). Notice that the objective function of $\SDP(V)$ is log-concave, so it can be optimized in polynomial time. It turns out that upon solving $\SDP(V)$, we can reduce to the case that the maximizer $X^*$ of $\SDP(V)$ satisfies $v_i^\dag X^* v_i = 1$ for all $i\in [n]$ (see \cref{eqn:alg-assumption}). This property simplifies matters enough that we will assume it for the rest of this section. 

The above property implies that 
$A\preceq I$ (see \cref{eqn:alg-optimality}), so we immediately get $\per(A)\leq 1$. Conversely, \cite{AGOS17, YP22} prove that 
\begin{equation}\label{eqn:per-exp}
    \per(A)\geq \frac{n!}{n^n}\cdot r(V)\geq e^{-n}\cdot r(V).
\end{equation}
Noticing that $\SDP(V)$ is a semidefinite relaxation of $r(V)$, a simple Gaussian rounding argument (\cite[Lemma 4.3]{YP22}, \cref{lem:alg-tight}) can be used to show 
\begin{equation}\label{eqn:gaussian-rounding}
    r(V)\geq e^{-\gamma n}\SDP(V) = e^{-\gamma n}.
\end{equation}
Putting these together, one gets
\begin{equation}\label{eqn:per-approx}
    e^{-(\gamma + 1)n}\leq \per(A)\leq 1,
\end{equation}
giving a $e^{-(\gamma + 1)n}$ approximation factor.

We remark that both sides of the above inequality can be tight, in particular the upper bound is tight for the identity matrix and the lower bound is tight for a certain family of low rank projection matrices (see \cite{AGOS17}). So one may expect that no improvement is possible along this line.

% The basis for our approach is the observation that we can refine the opposite estimate for all known examples where one of the bounds in \cref{eqn:per-approx} is tight -- for matrices close to $I$ (that achieve the upper bound) the basic Gaussian rounding used to prove $\|V\|_{2\to 0}^{2n}\geq e^{-\gamma n}$ is far from optimal, and for the low rank projection matrices achieving the lower bound, we are able to improve the upper bound from $1$ to about $e^{-n}$ (see proof of \cref{cor:proj}).

In our approach we exploit the fact these inequalities are tight for matrices of very different rank -- the known tight examples for the upper/lower bounds have very high/low rank respectively. In order to make this intuition concrete, we will use $\tr(A)$ as a smooth analogue of rank. Our main technical results are improvements to both sides of \cref{eqn:per-approx}.

\begin{lemma}[Improved Upper Bound]\label{lem:alg-upper}
    Let $\eps \in [0,1]$. Any matrix $0\preceq A\preceq I$ with $\tr(A)\leq (1-\eps)n$ satisfies \[\per(A)\leq \left(1-\frac{\eps^2}{20}\right)^n.\]
    % \[\per(A)\leq \det\left((2I-A)^{-1}\right)\leq \left(\frac{1}{2} + \frac{\tr(A)+\|A\|_F^2}{4n}\right)^n.\]
\end{lemma}
\begin{lemma}[Improved Lower Bound]\label{thm:rounding}
    Let $VV^\dag = A\preceq I$, and assume that the maximizer $X^*$ of $\SDP(V)$ satisfies $v_i^\dag X^* v_i = 1$ for all $i$. For any $0\leq \beta\leq 1$,
    \[\per(A)\geq e^{-n}\cdot r(V)\geq e^{-(\gamma + 1)n}\cdot \exp\left(n\cdot \left(\ln(1-\beta) + \frac{\beta}{1-\beta}\cdot \frac{\tr(A)}{n} - \frac{0.273\beta^2}{(1-\beta)^2}\cdot \frac{n}{\tr(A)}\right)\right).\]
\end{lemma}

Our proof of \cref{lem:alg-upper} is inspired by an identity for $\per(A)$ appearing in \cite{Bar20}. Our proof of \cref{thm:rounding} is based on an improved rounding procedure and is more technical, so we provide a proof overview in \cref{subsec:rounding-overview}. We can now state our algorithm.
\paragraph{Algorithm:} Given a PSD matrix $A=VV^\dagger$ where $v_1,\dots,v_n$ are rows of $V$, first reduce to the case that the maximizer $X^*$ of $\SDP(V)$ satisfies $v_i^\dagger X^*v_i=1$ for all $i$ (as described in \cref{eqn:alg-assumption}). Output $\left(1-\frac{\eps^2}{20}\right)^n$, where $\eps$ is defined by $\tr(A) = (1-\eps)n$.\\

We will use this algorithm in our proof of \cref{thm:upper}, which is straightforward to analyze when equipped with \cref{lem:alg-upper,thm:rounding}.
%can normalize $A$ to obtain $\tilde{A}$ such that the maximizer $X^*$ of $\SDP(\tilde{A})$ satisfies 
%\begin{equation}\label{eqn:alg-assumption}
%\end{equation}

\subsubsection{Hardness results}\label{sec:hardness-technical}
In this part we highlight the main technical contributions behind \cref{thm:per-thm}. Our proof broadly consists of two steps:

\begin{enumerate}
    \item Show that $r(V)$ does not admit a $e^{-\gamma n(1+\eps)}$-approximation algorithm.\label{item:h1}
    \item Give an approximation-preserving reduction from $r(V)$ to PSD permanents.\label{item:h2}
\end{enumerate}

We start by elaborating on \cref{item:h1}. In order to draw analogies to the existing hardness of approximation literature, we will first rephrase and generalize the optimization problem of $r(V)$. Let $\F\in \{\R, \C\}$ be a field. For a vector $x\in \F^n$ and $p\in \R-\{0\}$, define 
$$\norm{x}_p = \left(\E_i |x_i|^p\right)^{1/p}.$$ 
We will be particularly interested in the case that $p=0$, for which we define $\norm{x}_0 = \lim_{p\to 0}\norm{x}_p$. It is not too hard to see that for any vector $x$, $\norm{x}_0$ equals $\prod_{i\in [n]}|x_i|^{1/n}$, the geometric mean of the magnitude of the entries of $x$. Note that in the case that $p<1$, $\|\cdot\|_p$ is not a norm and not convex, but we will nevertheless refer to it as the $p$-norm.

% \textbf{Q: ok to still call them norms?} 

%    \[[x]_f^f := \E_{i\sim [n]}f(|x_i|),\]
 %   \[[x]_f := f^{-1}([x]_f^f) = f^{-1}(\E_{i\sim [n]}f(|x_i|).\]
   % We will refer to $[x]_f$ as the \emph{$f$-mean} of $x$. 
  %  More generally, for a random variable $X$ over $\F$ we define
  %  \[
 %   [X]_f^f = \E f(|X|),\quad [X]_f = f^{-1}(\E f(|X|)).
%    \]
%For a vector $x\in \R^n$ we write $$
% 
Given a matrix $A\in \F^{m\times n}$, the $p\to q$ ``norm'' of $A$ is defined as
$$ \|A\|_{p\to q} = \max_{x\in \F^n:\norm{x}_p=1} \norm{Ax}_q.$$

The connection of $\|A\|_{p\to q}$ to $r(V)$ is apparent: for any matrix $V\in \C^{n\times d}$, we have $$r(V) = \|V\|_{2\to 0}^{2n}.$$
Over the last decade there has been significant interest in designing approximation algorithms or proving hardness of approximation for matrix $p\to q$ norms for $p,q\geq 1$ \cite{BFHKSZ12,briet2015tight,bhattiprolu2018inapproximability}. Most notably, the $2\to 4$ norm has been shown to be closely related to the Unique games and the small set expansion conjectures \cite{BFHKSZ12}. To the best of our knowledge, the problem is not well-studied when $q<1$. We prove tight hardness of approximation (assuming $\PPP\neq \PNP$) for the $2\to q$ norm when $-1<q<2$. 

 For $q>-1$ let 
$$ \gamma_{\F,q} = \E_{g\sim \F\N(0,1)}[|g|^q]^{1/q}$$
be the  $q$-norm of a standard (real/complex) normal  random variable. 
Bhattiprolu, Ghosh,  Guruswami, Lee, Tulsiani \cite{bhattiprolu2018inapproximability} showed that for any $1\leq q<2$, and any $\eps>0$ it is $\PNP$-hard to approximate the $2\to q$ norm of a real $m\times n$ matrix better than $\gamma_{\R,p}+\eps$, matching known semidefinite relaxation-based approximation algorithms \cite{steinberg2005computation}.
In our main theorem, we build on their techniques and we extend their result to all $-1<q<2$.
\begin{theorem}[Main Technical Hardness Theorem]\label{thm:main-thm}
    Let $\F\in \{\R, \C\}$. For all  $-1 < q < 2$ and $\eps>0$, it is $\PNP$-hard to approximate $\|A\|_{2\to q}$ given a matrix $A\in \F^{m\times n}$ within a factor of $\gamma_{\F,q}+\eps$.
\end{theorem}

For the sake of completeness, in \cref{appendix:alg} we write down a semidefinite relaxation of $\|A\|_{2\to q}$ for all $-1<q<2$ and prove that it gives a $\gamma_{\F, q}$-approximation to $\|A\|_{2\to q}$, matching the above hardness result. As $r(V) = \|V\|_{2\to 0}^{2n}$, we also get that it is NP-hard to approximate $r(V)$ within a factor of $(\gamma_{\C, 0} + \eps)^{2n} = e^{-\gamma n (1+\eps)}$.

Next, we elaborate on \cref{item:h2} -- an approximation preserving reduction from $r(V)$ to $\per(A)$. Our main observation is that the permanent of a highly rank-deficient $n\times n$ PSD matrix $A=VV^\dagger$ is essentially (up to subexponential error) the same as $r(V)$. This is a consequence of Wick's formula (\cref{lem:cnd}), which allows us to view the permanent of a PSD matrix as a squared absolute moment of a complex multivariate Gaussian. As a result, we are able to use \cref{thm:main-thm} to prove \cref{thm:per-thm}, which we do in \cref{sec:permanent-proofs}. %We remark that the same techniques can be used to show that a degree $2n$ squared moment of a \emph{real} multivariate Gaussian cannot be approximated much better than $\gamma_{\R, 0}^{2n} = e^{-\gamma n}/2^n$.

\subsection{Overview of the proof of \cref{thm:rounding}}\label{subsec:rounding-overview}

Let us start by explaining the proof of \cref{eqn:gaussian-rounding}, which \cref{thm:rounding} improves upon. For any distribution $\mathcal{D}$ over $\C^n$ with $\E[\|x\|_2^2]=1$, we have the bound
\[r(V) = \max_{\|x\|_2 = 1}\prod_{i\in [n]}|\langle v_i, x\rangle|^2 \geq \left(\frac{\E_{x\sim \mathcal{D}}\prod_{i\in [n]}|\langle v_i, x\rangle|^{2/n}}{\E_{x\sim \mathcal{D}}\|x\|_2^2}\right)^{n} = \E_{x\sim \mathcal{D}}\left[\prod_{i\in [n]}|\langle v_i, x\rangle|^{2/n}\right]^n.\]

Using Jensen's inequality on the RHS, we get
\begin{equation}\label{eqn:gaussian-round}
    r(V)\geq \exp\left(\sum_{i\in [n]}\E_{x\sim \mathcal{D}}\ln|\langle v_i, x\rangle|^2\right).
\end{equation}

The basic Gaussian rounding scheme picks $x\sim \mathcal{D} = \C\N(0, X^*)$ (see \cref{def:gaussian} for a definition of the complex Gaussian distribution). Notice that for each $i$, $\langle v_i, x\rangle\sim \C\N(0, v_i^\dag X^*v_i) = \C\N(0, 1)$ by assumption of \cref{thm:rounding}. Since $\E_{y\sim \C\N(0, 1)}\ln|y|^2 = -\gamma$, we immediately get $r(V)\geq \exp(-n\gamma)$.

One can see that the analysis (in particular, the application of Jensen's inequality) is tight if $A = V = X^* = I$ for example. So to improve on \cref{eqn:gaussian-round}, we must use a different rounding algorithm. Our first observation is that in the special case that $A = V = I$ where $v_i = e_i$, we can get the optimal lower bound by sampling independent Rademachers $s_1,\ldots, s_n\sim \{\pm 1\}$, and setting $x =\sum_{i\in [n]} s_i v_i$. With this choice, $\E_{x}\ln|\langle v_i, x\rangle|^2 = \E_x \ln 1 = 0$, implying $r(V) \geq 1$.

One could try to use a similar rounding scheme in the more general case of \cref{thm:rounding}, i.e., when $A$ is close to $I$ in the sense that $\tr(A)\geq (1-\eps)n$. Unfortunately this strategy ends up failing, as when the $v_i$'s are not exactly orthogonal, there could be a nonzero probability that $\langle v_i, x\rangle = 0$, which would imply $\E_{x}\ln|\langle v_i, x\rangle|^2 = -\infty$. Note that if $A$ is close to $I$, this singularity is a very small probability event for most of the vectors $v_i$, so it is natural to try to avoid it by adding some noise to $x$. We do this by interpolating between the two rounding schemes. We pick a parameter $0<\beta<1$, and set $x = \sqrt{1-\beta} g + \sqrt{\beta}\sum_{i\in [n]}s_i v_i$, up to some normalization, where $g\sim \C\N(0, X^*)$.

On the technical side, this interpolation helps us analyze $\E_{x}\ln|\langle v_i, x\rangle|^2$ in terms of tractable quantities. We use a sharp bound on the expected log of the magnitude squared of a noncentral complex Gaussian (see \cref{lem:better-ln-estimate}): for any $c\in \C$,
\[\E_{g\sim \C\N(0, 1)}\ln|g + c|^2 \geq -\gamma + |c|^2 - \frac{|c|^4}{4}.\]
As a result of this inequality, when $\beta$ is bounded away from $1$, we can effectively bound $\E_{x}\ln|\langle v_i, x\rangle|^2$ using only the second and fourth moments of the random variable $\sum_{i\in [n]}s_i v_i$, which are tractable.

\subsection{Overview of the proof of \cref{thm:main-thm}}
As alluded to in \cref{sec:hardness-technical}, prior to our work, optimal hardness results for the $2\to q$ norm are already established for $q\geq 1$. Our first observation is that these results \cite{GRSW16,briet2015tight,bhattiprolu2018inapproximability} can be extended to all $-1 < q < 2$ (see \cref{thm:reduction}), or even more generally, to $2$-concave $f$-means (see \cref{def:f-mean}).

In particular, one can deduce the following theorem.

\begin{theorem}[Informal version of \cref{thm:reduction}]\label{thm:reduction-informal}
    Let $q< 2$. Assume that there is a family $\{E_k\}$ of $k\times d_k$ gadget matrices such that $\norm{E_k}_{2\to q} = 1$, but for all ``smooth'' unit vectors $x$ ($\|x\|_\infty\ll 1$), $\|E_kx\|_q \leq \gamma$. Then for all $\eps>0$, it is NP-Hard to distinguish between the following two cases given a matrix $A:\C^m\to \C^n$ with $\norm{A}_{2\to 2}\leq 1$.
    \begin{enumerate}
        \item Completeness: $\norm{A}_{2\to q}= 1$, or
        \item Soundness: $\norm{A}_{2\to q}\leq \gamma+\eps$.
    \end{enumerate}
\end{theorem}

It remains to construct an appropriate family of gadget matrices $\{E_k\}$. We will use the following family, which was suggested in \cite{briet2015tight}. For $k \geq 1$, define $E^{(\C)}_k \in \C^{4^k \times k}$ as the matrix whose rows consist of the members of $\frac{1}{\sqrt{k}}\cdot \{-1,+1, -i, +i\}^k$ ordered arbitrarily.

It remains to show that the matrices $E^{(\C)}_k$ satisfy the requirements of \cref{thm:reduction-informal} with $\gamma\approx \gamma_{\C, q}$. By construction, $\norm{E^{(\C)}_k}_{2\to q} = 1$.

To prove this, in \cref{lem:A1} we prove a Berry-Esseen type result for test functions of the form $|x|^q$ for $q\neq 0$ and $\log|x|$ otherwise, applied to a sum of independent random variables. In particular, the special case of interest to us (for \cref{thm:per-thm}) is $q=0$. In that case, the test function is $\log|x|$ which has a singularity at $x=0$, but we are nevertheless one can bound the right hand side of the lemma below by an arbitrarily small quantity as $\delta\to 0$.

\begin{lemma}[Informal version of \cref{lem:A2}]
% \label{lem:A2}
Let $-1 < q < 2$ and $0<\delta < 1$. For all ``smooth'' unit vectors $x$ with $\|x\|_\infty\leq \delta$,
    \[
    f(\|E_kx\|_q)  -\gamma_{\C, q}\lesssim 
    -\int_0^{\delta}  \min(0, f(u)) + \delta\cdot \left(\max(0, f(2\sqrt{\log(1/\delta)})) +2 f'(1)\right),
    \]

    where $f(x) = |x|^q$ for $q\neq 0$, and $f(x) = \log|x|$ for $q=0$.
\end{lemma}

\subsection{Future Directions}
The most exciting open problem is to determine the correct approximability for PSD permanents. Improving the hardness result seems to be out of reach of current techniques, but our ideas provide a clear path to improving the algorithmic result. Any significant improvement to the algorithm along the lines of our ideas would require significantly better versions of \cref{lem:alg-upper,thm:rounding}. In particular, \cref{lem:alg-upper} is currently the bottleneck to a better approximation ratio, specifically the $O(\eps^2)$ dependence. We conjecture that it can be improved to $O(\eps)$, which would yield better approximation ratios as a corollary.

Although we don't have concrete new applications of hardness of approximation of $\|A\|_{2\to q}$ for $q\neq 0$, we expect to find further applications of the machinery developed here in addressing counting and optimization of linear algebraic problems, e.g., in estimating mixed discriminant, sub-determinant maximization, Nash-welfare maximization, etc.

\subsection{Paper Organization}
In \cref{sec:prelims}, we present preliminary definitions and results that we will use. In \cref{sec:alg}, we prove \cref{thm:upper}. In \cref{sec:hardness}, we prove \cref{thm:per-thm,thm:main-thm} (with some components of the proof appearing in \cref{sec:appendix-A2,sec:appendix-berry-esseen}).

\section{Preliminaries}\label{sec:prelims}

\subsection{Generalized Means and Norms}

Although we mostly use $p$-norms that are defined using an expectation, it will be convenient to also define the counting version of $2$-norm, which we denote as $\ell_2$.
\begin{definition}[$\ell_2$-norm, Frobenius norm]
    Let $\F\in \{\R,\C\}$. For a vector $x\in \F^n$, define $\|x\|_{\ell_2} = \sqrt{\sum_{i\in [n]}|x_i|^2}$. For a matrix $A\in \F^{m\times n}$, define $\|A\|_F = \sqrt{\sum_{i,j\in [n]}|A_{i,j}|^2}$.
\end{definition}

We work with a generalization of $\|\cdot\|_p$ using the framework of ``$f$-means''.

\begin{definition}\label{def:f-mean}
    Let $f:\R_{\geq 0}\to \R$ be continuous and injective. Let $\F\in \{\R, \C\}$ be a field. For a vector $x\in \F^n$, define 
    \[[x]_f^f := \E_{i\sim [n]}f(|x_i|),\]
    \[[x]_f := f^{-1}([x]_f^f) = f^{-1}(\E_{i\sim [n]}f(|x_i|).\]
    We will refer to $[x]_f$ as the \emph{$f$-mean} of $x$. 
    More generally, for a random variable $X$ over $\F$ define
    \[
    [X]_f^f = \E f(|X|),\quad [X]_f = f^{-1}(\E f(|X|)).
    \]

    For a matrix $A\in \F^{n\times d}$, define 
    \[[A]_{2\to f} := \max_{x\in \F^n}\frac{[Ax]_f}{\|x\|_2}.\]
\end{definition}

$f$-means provide a convenient and unified way to talk about $\|\cdot\|_p$, even for the case of $p=0$.

\begin{observation}[Power Means]\label{obs:f-mean}
    For all $p\in \R$, define   

    \[f_p(x) = \begin{cases}
        x^p & \text{if }p>0,\\
        \log x& \text{if }p=0,\\
        -x^p& \text{if }p<0.
    \end{cases}\]

    Then, we have $[x]_{f_p} = \|x\|_p$ for any $p\in \R$.
\end{observation}

We note that the observation would still hold if we used $x^p$ instead of $-x^p$ in the third case, but it will be convenient for $f_p$ to always be an increasing function.

We will mostly be concerned with $f$-means that are dominated by the $2$-norm. This happens exactly when $f$ is $2$-concave:

\begin{definition}
        A function $f:\R_{> 0}\to \R$ is $2$-concave if $x\to f(\sqrt{x})$ is concave.
\end{definition}

\begin{example}
    For any $p\leq 2$, $f_p$ is 2-concave.
\end{example}
We make some useful observations about $2$-concave functions.
\begin{claim}\label{claim:2-concave}
    Let $f:\R_{> 0}\to \R$ be a $2$-concave increasing function. Then,
    \begin{enumerate}
        \item $[x]_f\leq \|x\|_2$ for any vector $x$. Equivalently, $[x]_f^f\leq f(\|x\|_2)$.
        \item $f'(x)\leq \frac{f'(y)}{y}\cdot x$ for $0<y\leq x$.
        \item $f(x)-f(y)\leq \frac{f'(y)}{2y}\cdot x^2$ for $0<y\leq x$.
    \end{enumerate}
\end{claim}
\begin{proof}
    \begin{enumerate}
        \item Using Jensen's inequality on the concave function $x\to f(\sqrt{x})$,
        \[[x]_f^f = \E f(|x_i|) = \E f\left(\sqrt{|x_i|^2}\right)\leq f(\sqrt{\E |x_i|^2}) = f(\|x\|_2).\]
        \item Follows from the fact that $f'(\sqrt{x}) = \frac{f(\sqrt{x})}{2\sqrt{x}}$ is increasing in $x$.
        \item Integrating the above from $y$ to $x$,
        \[f(x) - f(y)\leq \frac{f'(y)}{y}\cdot \frac{(x^2 - y^2)}{2}\leq \frac{f'(y)}{2y}\cdot x^2.\]
    \end{enumerate}
\end{proof}

One could ask when $[x]_f$ is a homogeneous function of $x$. It turns out that this is exactly when $[x]_f$ is a $p$-norm.

\begin{lemma}[\cite{hardy1952inequalities}]\label{lem:homogeneous}
    $[\cdot]_f$ is $1$-homogeneous (that is, $[ax]_f = a[x]_f$ for all scalars $a$ and vectors $x$) if and only if it equals $[\cdot]_{f_p} = \|\cdot\|_p$ for some $p\in \R$.
\end{lemma}

We will require another simple fact about $f$-means.

\begin{fact}\label{claim:tensor}
    Let $V$ be a matrix $V\in \F^{n\times d}$ and integer $k>0$. Then, 
    \[[V^{(k)}]_{2\to f} \coloneqq \left[\begin{bmatrix}
        V\\\vdots\\V
    \end{bmatrix}\right]_{2\to f} = [V]_{2\to f} ,\]
    where $V$ is copied $k$ times in the right hand side.
\end{fact}
\begin{proof}    
    For any vector $x\in \F^d$, consider the two vectors $z = Vx$ and $z^{(k)} = V^{(k)}x$. Note that a uniformly random entry of $z$ has the same distribution as a random entry of $z^{(k)}$, so $[z]_f = [z^{(k)}]_f$. The claim follows from the definition of $[\cdot]_{2\to f}$.
\end{proof}

% \begin{fact}
%     \end{fact}
% Equivalently, we have 
% \[
% f^{-1}(\E_{i \in [n]} f(c x_i))
% \]

\subsection{Gaussians}

We will consider both real and complex Gaussians.

\begin{definition}\label{def:gaussian}
    Let $\F\in \{\R, \C\}$. For the $n\times n$ identity matrix $I_n$, $\F\N(0, I_n)$ is defined to be the distribution over vectors $x\in \F^n$ given by the density function 
    \[p_\F(x) = \begin{cases}
        (2\pi)^{-n/2}\cdot \exp\left(-\|x\|_{\ell_2}^2/2\right) &\text{ if }\F = \R,\\
        \pi^{-n}\exp(-\|x\|_{\ell_2}^2)&\text{ if }\F = \C.
    \end{cases}\]

    More generally, given a positive semidefinite covariance matrix $\Sigma= A A^\dag$ for $A\in \F^{n\times d}$, define $\F\N(0, \Sigma)$ to be distributed as $Ax$, where $x\sim \F\N(0, I_d)$. We will sometimes use $\N$ to denote $\R\N$.
\end{definition}

More concretely, a complex Gaussian $g\sim \C\N(0, 1)$ can be sampled by sampling its real and imaginary parts independently from $\N(0, 1/2)$. There are formulas for the moments of univariate real and complex Gaussians in terms of the Gamma function.

\begin{definition}
    For any $p\in \R$ and $\F\in \{\R, \C\}$, define $\gamma_{\F, p} = [g]_{f_p}$, where $g$ is a random variable distributed as $\F\N(0, 1)$.
\end{definition}

\begin{fact}\label{fact:gamma-values}
    For any $p\in (-1,\infty)- \{0\}$,
    \[\gamma_{\R, p}^p = \E_{g\sim \N(0, 1)}[|g|^p] = \frac{2^{p/2}\cdot \Gamma\left(\frac{p+1}{2}\right)}{\sqrt{\pi}},\]
    and for any $p\in (-2,\infty)- \{0\}$,
    \[\gamma_{\C, p}^p  = \E_{g\sim \\CN(0, 1)}[|g|^p] = \Gamma\left(\frac{p}{2}+1\right).\]

    In particular, this implies
    \[\gamma_{\R, 1} = \sqrt{\frac{2}{\pi}},\quad \gamma_{\C, 1} = \sqrt{\frac{\pi}{2}}, \quad \gamma_{\R, 0} = \lim_{p\to 0}\gamma_{\R, p} = \sqrt{\frac{e^{-\gamma}}{2}},\quad \gamma_{\C, 0} = \lim_{p\to 0}\gamma_{\C, p} = \sqrt{e^{-\gamma}}.\]
\end{fact}

\begin{fact}[Moment Generating Function of $|g|^2$]\label{fact:mgf}
    Let $g\sim \C\N(0, 1)$. For any $t<1$, $\E[e^{t|g|^2}] = (1-t)^{-1}$.
\end{fact}

We will need sharp bounds on the expected value of $\ln|g + c|^2$ for Gaussian $g$ and fixed $c$. First, we prove an estimate on the exponential integral function.

\begin{fact}\label{fact:ei-identity}
For $x \geq 0$ it holds  that
\[
\Ei(-x) = \int_{-\infty}^{-x}\frac{e^t}{t}dt = \gamma + \ln(x) + \sum_{k=1}^\infty \frac{(-x)^k}{k \cdot k!}\leq \gamma + \ln(x) - x + \frac{x^2}{4}.
\]
\end{fact}
\begin{proof}
    The identity is due to Equation 5.1.11 in \cite{abramowitz1948handbook,}. For the inequality, we must show that the function
    \[f(x) = \sum_{k = 3}^\infty \frac{(-x)^k}{k\cdot k!}\]
    is nonpositive for $x\geq 0$. To do this, observe first that $f(0) = 0$, and 
    \begin{align*}
        f'(x) &= \sum_{k = 3}^\infty \frac{(-1)^k\cdot x^{k-1}}{k!}\\
        &= \frac{1}{x}\sum_{k=3}^\infty \frac{(-x)^{k}}{k!}\\
        &= \frac{e^{-x} - \left(1 - x + \frac{x^2}{2}\right)}{x}\\
        &\leq 0.\tag{$e^{-x}\leq 1-x+x^2/2$ for $x\geq 0$}
    \end{align*}
Therefore, $f(x)\leq 0$ for $x\geq 0$.
\end{proof}

\begin{lemma}\label{lem:better-ln-estimate}
    Let $c\in \C$. Then, $\E_{g\sim \C\N(0, 1)}[\ln|g+c|^2]\geq -\gamma + |c|^2 - |c|^4/4$.
\end{lemma}
\begin{proof}
    Define $x = |c|^2$. By \cite[Eqn. 35, Thm. 1]{moser2020} we have the identity
    \[\E_{g\sim \C\N(0,1)}[\ln|g+c|^2] = - \ln(x) - \Ei(x).\]
    By \cref{fact:ei-identity}, this is at least $-\gamma + x - x^2/4$, as desired.
\end{proof}

\subsection{Permanent}
	For a matrix $A\in \C^{n\times n}$, its permanent is defined as
	\[ \per(A):=\sum_{\sigma\in S_n}\prod_{i=1}^n A_{i,\sigma(i)}. \]

    On the domain of positive semidefinite matrices, the permanent has some nice properties. For example, it is monotone w.r.t. the Loewner order.
 	
	\begin{lemma}[e.g., \cite{AGOS17}]
		\label{lem:permanent-monotone}
		If $A, B\in \C^{n\times n}$ are hermitian and $A\succeq B\succeq 0$, then
		\[ \per(A)\geq \per(B). \]
	\end{lemma}
	\begin{proof}[Proof Sketch]
		The statement of the lemma follows, because $A\succeq B\succeq 0$ implies that $A^{\otimes n}\succeq B^{\otimes n}\succeq 0$. So, if $1_{S_n}$ is the indicator vector of all permutations in $\R^{n\otimes n}$,
		$$ \per(A) = \frac{1}{n!} 1_{S_n}^\dag A^{\otimes n}1_{S_n} \geq \frac{1}{n!} 1_{S_n}^\dag B^{\otimes n}1_{S_n} = \per(B)$$
		as desired.
	\end{proof}

\begin{lemma}\label{lem:cnd}
    For any PSD matrix $V V^\dag$ with $V\in \R^{n\times d}$, we have
    \[\E_{x\sim \N(0, I)}\left[\prod_{i\in [n]}|\langle v_i, x\rangle|^2\right] = c_{n,d}^\R\cdot \E_{x\in \R^d, \|x\|_{2} = 1}\left[\prod_{i\in [n]}|\langle v_i, x\rangle|^2\right].\]
    
    For any PSD matrix $V V^\dag$ with $V\in \C^{n\times d}$, we have
    \[\per(V V^\dag) = \E_{x\sim \C\N(0, I)}\left[\prod_{i\in [n]}|\langle v_i, x\rangle|^2\right] = c_{n,d}^\C\cdot \E_{x\in \C^d, \|x\|_{2} = 1}\left[\prod_{i\in [n]}|\langle v_i, x\rangle|^2\right].\]

    Here, $v_1,\ldots, v_n$ are the rows of $V$. The proportionality constants above are defined by
    \[c_{n,d}^\R = \frac{\Gamma(n+d/2)}{ \Gamma(d/2)\cdot (d/2)^n},\qquad c_{n,d}^\C = \frac{(d+n-1)!}{(d-1)!\cdot d^n}.\]
\end{lemma}
\begin{proof}
    The first equality in the second conclusion follows from Isserlis' theorem/Wick's formula (see 3.1.4 in \cite{barvinok2016combinatorics}).

    We prove the other equality in the real case, but the complex case can be proved similarly. Observe that 
    \begin{align*}
        \E_{x\sim \N(0, I)}\left[\prod_{i\in [n]}|\langle v_i, x\rangle|^2\right]& = \E_{x\sim \N(0, I)} \norm{x}_{2}^{2n}\cdot \left[\prod_{i\in [n]}\left|\left\langle v_i, \frac{x}{\norm{x}_{2}}\right\rangle\right|^2\right]\\
        &=\E_{x\sim \N(0, I)} \norm{x}_{2}^{2n}\cdot \E_{x\in \R^n,\norm{x}_{2}=1}\left[\prod_{i\in [n]}\left|\left\langle v_i, x\right\rangle\right|^2\right]\\
        &=d^{-n}\E_{x\sim \N(0, I)} \norm{x}_{\ell_2}^{2n}\cdot \E_{x\in \R^n,\norm{x}_{2}=1}\left[\prod_{i\in [n]}|\left\langle v_i, x\right\rangle|^2\right]
    \end{align*}
    The first identity uses that for $x\sim \N(0, I)$, $\|x\|_2$ is independent from $x/\|x\|_2$. The second identity uses that fact that if $x\sim \N(0,I)$, then $x/\|x\|_2$ is distributed uniformly on a sphere of radius $\norm{x}_{2}$.
    To conclude the proof, notice that $\E_{x\sim \N(0, I)} \norm{x}_{\ell_2}^{2n}$ is the $n^{th}$ moment of a chi-squared random variable with $d$-degrees of freedom, which is $2^n\frac{\Gamma(n+d/2)}{ \Gamma(d/2)}$.
\end{proof}

We will also require the following formula for the permanent of the sum of two matrices.

\begin{lemma}[{\cite[Page 2]{percus2012combinatorial}}]\label{lem:per-sum}
    For any two matrices $A,B\in \C^{n\times n}$,
    \[\per(A + B) = \sum_{S, T\subseteq [n], |S| = |T|}\per(A_{S, T})\cdot \per(B_{\bar{S}, \bar{T}}).\]
    When $A=I$, this simplifies to
    \[\per(I + B) = \sum_{S\subseteq [n]}\per(B_{S, S}).\]
    Above, $A_{S,T}$ is the $|S|\times |T|$ submatrix of $A$ containing rows only in $S$ and columns only in $T$.
\end{lemma}

\section{Algorithm}\label{sec:alg}

We start by expanding on the basic setup of the algorithms of \cite{AGOS17, YP22}, which we briefly introduced in \cref{sec:alg-technical}. After this, we will show how \cref{lem:alg-upper,thm:rounding} imply \cref{thm:upper}. Later on, in \cref{subsec:upper,subsec:lower} respectively, we prove \cref{lem:alg-upper,thm:rounding}.

Let $A = VV^\dag$ be the PSD matrix whose permanent we wish to compute. Let $v_1,\ldots, v_n\in \CC^n$ be the rows of $V$, so by \cref{lem:cnd},
\[\per(A) = \E_{x\sim \C\N(0, I)}\left[\prod_{i\in [n]}|\langle x, v_i\rangle|^2\right].\]

Recall the log-concave maximization problem $\SDP(V)$ we associated with this problem:
\[\SDP(V) = \max_{X: X \succeq 0, \tr(X) = n} \prod_{i\in [n]}  v_i^\dag X v_i.\]
Let $X^*$ be the optimal solution to $\SDP(V)$. Note that $X^*$ can be found efficiently. It will be convenient to make a simplification to our problem. We will replace the matrix $A$ by $\tilde{A} = D^{-1/2} A D^{-1/2}$, where $D$ is a positive semidefinite diagonal matrix defined as $D_{i, i} = v_i^\dag X^* v_i$. Since $D$ is diagonal,
$$\per(A) = \per(\tilde{A})\cdot \per(D)= \per(\tilde{A})\cdot \SDP(V),$$ so it suffices to approximate $\per(\tilde{A})$ instead of $\per(A)$. Writing $\tilde{A} = \tilde{V}\tilde{V}^\dagger$ for $\tilde{V} = D^{-1/2}V$, we can see that the objective functions of $\SDP(V)$ and $\SDP(\tilde{V})$ are positive scalar multiples of each other, so $\SDP(\tilde{V})$ is also maximized by $X^*$. Note that $\tilde{A}$ enjoys the additional property $\tilde{v}_i^\dag X^*\tilde{v}_i= 1$ for all $i\in [n]$, where $\tilde{v_i} = D^{-1/2} v_i$.
Replacing $A$ by $\tilde{A}$, we will henceforth assume that the maximizer $X^*$ of $\SDP(V)$ satisfies 
\begin{equation}\label{eqn:alg-assumption}
    v_i^\dag X^* v_i = 1 \text{ for all }i\in [n].
\end{equation}
In particular, this implies $\SDP(V) = 1$. Under this assumption, $A$ satisfies an important property.

\begin{claim}\label{eqn:alg-optimality}
    We have $A\preceq I$.
    % The two following statements hold:
    % \begin{enumerate}
    %     \item $A\preceq I$, and
    %     \item $V^\dag V X^* = X^*$.
    % \end{enumerate}
\end{claim}
\begin{proof}
    Let $f(X) = \prod_{i\in [n]}v_i^\dag X v_i$ be the objective function of $\SDP(V)$. We can compute $$\nabla(\ln f)(X) = \sum_{i\in [n]}\frac{v_i v_i^\dag}{v_i^\dag X v_i}.$$ In particular, by \cref{eqn:alg-assumption}, $\nabla(\ln f)(X^*) = \sum_{i\in [n]} v_i v_i^\dag = V^\dag V$. 

    The optimality conditions for $X^*$ imply that for all symmetric matrices $M$ with $\tr(M) = 0$ and $W_{-} \subseteq \range(X^*)$ it holds that
    \[
    \langle V^\dag V, M \rangle = \langle \nabla (\ln f)(X^*), M \rangle \leq 0.
    \]
    Here, $W_{-}$ denotes the vector space spanned by the negative eigenvectors of $M$. Now, let $Q\succeq 0$ be an arbitrary PSD matrix, and set $M = Q - \frac{\tr(Q)}{n}X^*$. $M$ satisfies both the conditions above, and therefore we have
    \begin{align*}
        0&\geq \langle V^\dag V, M\rangle\\
        &= \langle V^\dag V, Q\rangle - \frac{\tr(Q)}{n}\langle V^\dag V, X^*\rangle\\
        &= \langle V^\dag V, Q\rangle - \frac{\tr(Q)}{n} \sum_{i\in n}v_i^\dag X^* v_i\\
        &= \langle V^\dag V, Q\rangle - \tr(Q).\tag{by \cref{eqn:alg-assumption}}
    \end{align*}

    In other words, $\langle V^\dag V, Q\rangle\leq \tr(Q)$ for all $Q\succeq 0$, implying $V^\dag V\preceq I$. Therefore $A = VV^\dag \preceq I$.
    
\end{proof}
\cref{eqn:alg-optimality} immediately implies $\per(A)\leq 1$. In \cite{AGOS17, YP22}, the authors prove the complimentary inequalities
\begin{equation}\label{eqn:alg-lb}
 \per(A)\geq \frac{n!}{n^n}\cdot r(V)\geq \exp(-\gamma n)\cdot \frac{n!}{n^n}\cdot \SDP(V) = \exp(-\gamma n)\cdot \frac{n!}{n^n} \gtrsim \exp(-(\gamma + 1)n),
\end{equation}
and together, the two inequalities above provide a $e^{-(1+\gamma)n}$ approximation for $\per(A)$.

Recall \cref{lem:alg-upper,thm:rounding}, which (under \cref{eqn:alg-optimality}) improve the above inequalities to
\begin{equation}\label{eqn:alg-inequalities}
    e^{-(\gamma+1)n}\cdot \exp\left(n\cdot \ell\left(\frac{\tr(A)}{n}\right)\right)\leq \per(A)\leq \exp\left(n\cdot r\left(\frac{\tr(A)}{n}\right)\right).
\end{equation}
Here, $\ell(x) = \max_{0\leq\beta\leq 1}\ln(1-\beta) + \frac{\beta x}{(1-\beta)}- \frac{0.273\beta^2}{(1-\beta)^2x}$, and $r(x) = \ln\left(1-\frac{(1-x)^2}{20}\right)$. We are now ready to prove \cref{thm:upper}.

\begin{proof}[Proof of \cref{thm:upper}]
Let $A = VV^\dag \succeq 0$, where $V$ has rows $v_1,\ldots, v_n$. Our algorithm will first solve $\SDP(V)$ and use \cref{eqn:alg-assumption} and \cref{eqn:alg-optimality} to reduce to the case that $0\preceq A\preceq I$ and $v_i^\dag X^* v_i=1$ for all $i$, where $X^*$ is the optimal solution to $\SDP(V)$. We will then output $\exp\left(n\cdot r\left(\frac{\tr(A)}{n}\right)\right)=\left(1-\frac{\eps^2}{20}\right)^n$.

\cref{eqn:alg-inequalities} implies that the approximation factor of this algorithm is at least $e^{-(\gamma + 1 - \alpha)n}$, where $\alpha$ is the minimum value of $r(x) - \ell(x)$ over all $x\in [0,1]$. Write $\ell(x)\geq \ell'(x) := \max(0, \ln(1-\beta^*) + \frac{\beta^* x}{(1-\beta^*)}- \frac{0.273(\beta^*)^2}{(1-\beta^*)^2x}$ for $\beta^* = 0.34$. One can numerically determine that $\alpha\geq \min_{0\leq x\leq 1}r(x) - \ell'(x)\geq 10^{-4}$.
\end{proof}

\subsection{Proof of \cref{lem:alg-upper}}\label{subsec:upper}

We first prove an inequality that we will require. The proof of this inequality is inspired by an identity of Barvinok \cite{Bar20}.

\begin{lemma}\label{lem:per-cond-number-bound}
    For any matrix $0\preceq B\prec I$, $\per(I + B)\leq \det\left(\left(I-B\right)^{-1}\right)$.
\end{lemma}
\begin{proof}
    Write $B = VV^\dag$ for $V\in \C^{\times n}$. Let $v_1,\ldots, v_n$ be the rows of $V$.
    \begin{align*}
        \per(I + B) &= \sum_{S\subseteq [n]}\per(B_{S, S})\tag{\cref{lem:per-sum}}\\
        &= \sum_{S\subseteq [n]}\E_{g\sim \C\N(0, I)}\left[\prod_{i\in S}|\langle v_i, g\rangle|^2\right]\tag{\cref{lem:cnd}}\\
        &= \E_{g\sim \C\N(0, I)}\left[\prod_{i\in [n]}\left(1 + |\langle v_i, g\rangle|^2\right)\right]\\
        &\leq \E_{g\sim \C\N(0, I)}\left[\prod_{i\in [n]}e^{|\langle v_i, g\rangle|^2}\right]\tag{$1+x\leq e^x$ for all $x$}\\
        &=\E_{g\sim \C\N(0, I)}\left[\exp\left({g^\dag V^\dag V g}\right)\right].
    \end{align*}
    Let $\sigma_1,\ldots, \sigma_n$ be the eigenvalues of $V^\dag V$. Since $g$ is invariant under unitary transformations, we can rotate $g$ into the eigenbasis of $V^\dag V$ to get that
    \begin{align*}
        \per(I+B)&\leq\E_{g\sim \C\N(0, I)}\left[\exp\left(\sum_{i\in [n]}\sigma_i |g_i|^2\right)\right]\\
        &= \prod_{i\in [n]}\E_{g\sim \C\N(0, 1)}\left[e^{\sigma_i |g|^2}\right]\tag{Independence of $g_i$}\\
        &= \prod_{i\in [n]}\frac{1}{1-\sigma_i}.\tag{\cref{fact:mgf}}
    \end{align*}

    Noting that the eigenvalues of $V^\dag V$ match those of $B=VV^\dag$, this is equal to $\det\left((I-B)^{-1}\right)$.
\end{proof}
Now, we are ready to prove \cref{lem:alg-upper}. Let $0\preceq A\preceq I$ be a matrix with $\tr(A)\leq (1-\eps)n$. Let $0\leq \lambda_1\leq\ldots\leq\lambda_n\leq 1$ be the eigenvalues of $A$, and let $v_1,\ldots, v_n$ be the corresponding eigenvectors. Let $t\in (1/2, 1]$ be a parameter we will set later, and let $i_t$ be the smallest index $i$ such that $\lambda_i > t$. For any parameter $t\in (1/2, 1]$, we can write
\[A\preceq tI + \sum_{i\geq i_t}(\lambda_i - t)v_i v_i^\dag = t\cdot \left(I + \sum_{i\geq i_t}\frac{\lambda_i - t}{t}v_i v_i^\dag\right).\]
Write $B = \sum_{i\geq i_t}\frac{\lambda_i - t}{t}v_i v_i^\dag$. Since $t > 1/2$, $B\prec I$, so it satisfies the conditions of \cref{lem:per-cond-number-bound}.
\begin{align*}
    \per(A)&\leq t^n\cdot \per(I+B)\tag{\cref{lem:permanent-monotone}}\\
    &\leq \frac{t^n}{\det(I-B)}.\tag{\cref{lem:per-cond-number-bound}}
\end{align*}
We now pick $t = 1-\eps/5$. For this choice of $t$, we must have $i_t\geq \frac{\eps n}{2}$, since otherwise, $\tr(A)\geq t\cdot (n - i_t)\geq (1-\eps/5)\cdot (1-\eps/2)n > (1-\eps)n$ contradicts the fact that $\tr(A)\leq (1-\eps)n$. We compute
\[\det(I - B) = \prod_{i\geq i_t}\left(1-\frac{\lambda_i - t}{t}\right) = \prod_{i\geq i_t}\left(2-\frac{\lambda_i}{t}\right)\geq \left(2-\frac{1}{t}\right)^{n-i_t}.\]
Plugging in the definition of $t$ and our lower bound on $i_t$, this is at least
\[\left(2-\frac{1}{1-\eps/5}\right)^{(1-\eps/2)n} = \left(\frac{1-2\eps/5}{1-\eps/5}\right)^{(1-\eps/2)n}.\]
We now have our upper bound on $\per(A)$:
\[\per(A)\leq \left(1-\eps/5\right)^n\cdot \left(\frac{1-\eps/5}{1-2\eps/5}\right)^{(1-\eps/2)n} = \left(\frac{(1-\eps/5)\cdot (1-\eps/5)^{1-\eps/2}}{(1-2\eps/5)^{1-\eps/2}}\right)^n.\]
To complete the proof, we use that $\frac{(1-\eps/5)\cdot (1-\eps/5)^{1-\eps/2}}{(1-2\eps/5)^{1-\eps/2}}\leq 1-\frac{\eps^2}{20}$ for all $\eps \in [0,1]$.

\subsection{Proof of \cref{thm:rounding}}\label{subsec:lower}

By \cref{eqn:alg-lb}, it suffices to prove a lower bound on $r(V)$. Let $X^*$ be the optimal solution to $\SDP(V)$. Consider the following randomized rounding scheme to a solution of $r(V)$: sample $g\sim \C\N(0, X^*)$ and $s_i\sim \{z\in \C : |z|=1\}$ independently for all $i\in [n]$. Let $x = \sqrt{1-\beta}g + \sqrt{\frac{\beta n}{\tr(A)}}\sum_{i\in [n]}s_i v_i$. We will use the bound
\[r(V)^{1/n}=\max_{\|x\|_2^2 = 1}\prod_{i\in [n]}|\langle v_i, x\rangle|^{2/n} \geq \frac{\E_{x}[\prod_{i\in n}|\langle v_i, x\rangle|^{2/n}]}{\E_x[\|x\|_2^2]}.\]
First we compute the denominator. 
\begin{align*}
 n\cdot \E[\norm{x}_2^2] &= n\cdot \E[\left\|x\right\|_2^2]\\
 &= \E\left[\left\|\sqrt{1-\beta} g + \sqrt{\frac{\beta n}{\tr(A)}}\sum_{i\in [n]}s_i v_i\right\|_{\ell_2}^2\right]\\
 &= (1-\beta)\E[\|g\|_{\ell_2}^2] + \frac{\beta n}{\tr(A)}\E\left[\sum_{i,j}s_i s_j \langle v_i, v_j\rangle\right] + \sqrt{\frac{\beta n}{\tr(A)}(1-\beta)}\E\left[\left\langle g, \sum_i s_i v_i\right\rangle\right]\\
 &= (1-\beta) \E[\|g\|_{\ell_2}^2] + \frac{\beta n}{\tr(A)}\sum_{i}\|v_i\|_{\ell_2}^2\tag{Independence}\\
 &= (1-\beta) \cdot \tr(X^*) + \frac{\beta n}{\tr(A)}\cdot \sum_{i\in [n]}\|v_i\|_{\ell_2}^2\tag{$g\sim \C\N(0, X^*)$, definition of $\|\cdot\|_{\ell_2}$}\\
 &= n.\tag{$\tr(X^*)=n$, $\sum_{i\in [n]}\|v_i\|_{\ell_2}^2 = \tr(VV^\dag) = \tr(A)$}
\end{align*}

So, $\E[\|x\|_2^2] = 1$. It remains to lower bound the numerator. We start by applying Jensen's inequality to get
\begin{equation}\label{eqn:9}
    \E_{x}\left[\prod_{i\in n}|\langle v_i, x\rangle|^{2/n}\right] \geq \exp\left(\frac{1}{n}\sum_{i\in [n]}\E_{x}[\ln|\langle v_i, x\rangle|^2]\right).
\end{equation}

We will bound each of the terms inside the sum. Fix some $i\in [n]$, and let $y_i = \sqrt{\frac{n}{\tr(A)}}\sum_{j\in [n]}s_j \langle v_i, v_j\rangle$ and $z_i = \langle g, v_i\rangle$, so $\langle v_i, x\rangle = \sqrt{1-\beta}z_i + \sqrt{\beta} y_i$. Notice that $z_i\sim \C\N(0, v_i^\dag X^* v_i) = \C\N(0, 1)$ by assumption. Let us bound

\begin{align*}
    \E[\ln|\langle v_i, x\rangle|^2]&= \E[\ln|\sqrt{1-\beta}z_i + \sqrt{\beta}y_i|^2] \\
    &= \ln(1-\beta) + \E\left[\ln\left|z_i + \sqrt{\frac{\beta}{1-\beta}}y_i\right|^2\right]\\
    &\geq -\gamma + \ln(1-\beta) + \frac{\beta}{1-\beta}\E[|y_i|^2] - \frac{\beta^2}{4(1-\beta)^2}\E[|y_i|^4].\tag{\cref{lem:better-ln-estimate}, $z_i\sim \C\N(0, 1)$ and is independent of $y_i$}\\
\end{align*}
We bound the second and fourth moments of $y_i$ using the below claim, whose proof we defer to \cref{sec:ln-bound}. 

\begin{claim}\label{claim:ln-bound}
    For all $i\in [n]$,
    \[\E[|y_i|^2] = \frac{n}{\tr(A)}v_i^\dag V^\dag V v_i,\]
    % \[b_i = \frac{3n^2}{10(1-\beta)^2\tr(A)^2}\left(2\left(\sum_{j\in [n]}|\langle v_i, v_j\rangle|^2\right)^2 - \sum_{j\in [n]}|\langle v_i, v_j\rangle|^4  \right).\]
    \[\E[|y_i|^4] \leq \frac{1.09 n^2}{\tr(A)^2}\|v_i\|_{\ell_2}^2.\]
\end{claim}
Plugging in the bounds from \cref{claim:ln-bound} and summing over all $i$, we get

\begin{align*}
    \frac{1}{n}\sum_{i\in [n]}\E_x[\ln|\langle v_i, x\rangle|^2]&\geq -\gamma + \ln(1-\beta) + \frac{\beta}{(1-\beta)\tr(A)}\sum_{i\in [n]}v_i^\dag V^\dag V v_i - \frac{0.273\beta^2n}{(1-\beta)^2\tr(A)^2}\sum_{i\in [n]}\|v_i\|_{\ell_2}^2\\
    &= -\gamma + \ln(1-\beta) + \frac{\beta}{1-\beta}\cdot \frac{\|A\|_F^2}{\tr(A)} - \frac{0.273\beta^2}{(1-\beta)^2}\cdot \frac{n}{\tr(A)}\\
    &\geq -\gamma + \ln(1-\beta) + \frac{\beta}{1-\beta}\cdot {\frac{\tr(A)}{n}} - \frac{0.273\beta^2}{(1-\beta)^2}\cdot \frac{n}{\tr(A)}\tag{$\|A\|_F^2\geq \frac{\tr(A)^2}{n}$ by Jensen's inequality}
\end{align*}
This completes the proof of \cref{thm:rounding}.

\subsubsection{Proof of \cref{claim:ln-bound}}\label{sec:ln-bound}

\begin{proof}

    We can directly compute 

    \begin{align*}
        \frac{\tr(A)}{n}\E[|y_i|^2] &= \sum_{j, k\in [n]}\E[s_j \overline{s_k}]\cdot \langle v_i, v_j\rangle \overline{\langle v_i, v_k\rangle}\\ 
        &= \sum_{j\in [n]} |\langle v_i, v_j\rangle|^2 = v_i^\dag V^\dag V v_i\tag{$s_j$ is independent from $s_k$ for $j\neq k$}
    \end{align*}
    Similarly,

    \begin{align*}
        \frac{\tr(A)^2}{n^2}\E[|y_i|^4] &= \sum_{j, k, l, m\in [n]}\E[s_j s_k \overline{s_l s_m}]\langle v_i, v_j\rangle \langle v_i, v_k\rangle \overline{\langle v_i, v_l\rangle} \overline{\langle v_i, v_m\rangle} \\
        &= \sum_{j\in [n]}|\langle v_i, v_j\rangle|^4 + 2\sum_{j\neq k}|\langle v_i, v_j\rangle\langle v_i, v_k\rangle|^2\\
        %&= 2\sum_{j, k\in [n]}|\langle v_i, v_j\rangle\langle v_i, v_k\rangle|^2 - \sum_{j\in [n]}|\langle v_i, v_j\rangle|^2\\
        &= 2\left(\sum_{j\in [n]}|\langle v_i, v_j\rangle|\right)^2 - \sum_{j\in [n]}|\langle v_i, v_j\rangle|^4\\
        &= 2(v_i^\dag V^\dag V v_i)^2 - \sum_{j\in [n]}|\langle v_i, v_j\rangle|^4\\
        &\leq 2\|v_i\|_{\ell_2}^4 - \|v_i\|_{\ell_2}^8\tag{$V^\dag V\preceq I$, since $A = VV^\dag \preceq I$}\\
        &\leq 1.09\cdot \|v_i\|_{\ell_2}^2\tag{$x^2 - x^4\leq 1.09 x$ for $x\geq 0$}
    \end{align*}

    The second equality is because $\E[s_j s_k \overline{s_l s_m}] = 0$ unless each index appears an equal number of times in $\{j, k\}$ and $\{l, m\}$. 
\end{proof}

\section{Hardness of Approximation}\label{sec:hardness}

As mentioned in \cref{sec:technical-contributions}, we will first prove \cref{thm:main-thm}. Later on, in \cref{sec:permanent-proofs}, we will use \cref{thm:main-thm} to prove \cref{thm:per-thm} using an approximation-preserving reduction to the permanent problem.

Our first result is a general inapproximability result for the $f$-mean version of $\|A\|_{2\to q}$ that is dependent on an appropriate family of gadgets $\{E_k\}$ as defined below.

\begin{theorem}\label{thm:reduction}
    Let $\F\in \{\R, \C\}$, and let $f:\R_{\geq 0}\to \R$ be a continuous increasing $2$-concave function such that $\lim_{x\to\infty} \frac{f(x)}{x^2}=0$. 
    %that induces a homogeneous mean $[\cdot]_f$. 
    Let $\delta,\gamma > 0$. Assume that for all $k$, there is a matrix $E_k:\F^k\to \F^{d_k}$ satisfying the following:
    \begin{enumerate}
        % \item\label{cond:orth} $E_k^\dag E_k = \frac{d_k}{k}\cdot I_k$,
        \item\label{cond:norm} $\norm{E_k}_{2\to 2} = 1$.
        \item\label{cond:boolean} The entries of $E_k$ have magnitude equal to $\frac{1}{\sqrt{k}}$.
        \item\label{cond:sparse} for all vectors $x\in \F^k$ with $\|x\|_{\infty}\leq \delta \cdot \|x\|_{\ell_2}$, $[E_kx]_f\leq  \gamma\cdot \|x\|_{2}$.
    \end{enumerate}
    
    Then for all $\eps>0$, it is NP-Hard to distinguish between the following two cases given a matrix $A:\F^m\to \F^n$ with $\norm{A}_{2\to 2}\leq 1$.
    \begin{enumerate}
        \item Completeness: $[A]_{2\to f}= 1$, or
        \item Soundness: $[A]_{2\to f}\leq \gamma+\eps$.
    \end{enumerate}
\end{theorem}

The proof of the above theorem is in \cref{sec:reduction} and closely follows the arguments used in \cite{briet2015tight, bhattiprolu2018inapproximability}. In order to instantiate it, we will have to construct a family of gadgets $\{E_k\}$ that have $\|E_k\|_{2\to 2} = 1$, but at the same time have small $2\to f$-norm when restricted to ``smooth'' vectors.

\begin{definition}
        For $k \geq 1$, let us define
        $E^{(\R)}_k \in \R^{2^k \times k}$ as the matrix whose rows consist of the members of $\frac{1}{\sqrt{k}}\cdot \{-1,+1\}^k$ ordered arbitrarily. Similarly, we define $E^{(\C)}_k \in \C^{4^k \times k}$ as the matrix whose rows consist of the members of $\frac{1}{\sqrt{k}}\cdot \{-1,+1, -i, +i\}^k$ ordered arbitrarily.
\end{definition}
Observe that these matrices are normalized so that $\|E_k^{(\F)}\|_{2\to 2} = 1$. The following Lemma shows that condition 3 of \cref{thm:reduction} is satisfied with $\gamma\approx \gamma_{\F, p}$.

\begin{lemma}
\label{lem:A2}
Let $\F \in \{\R, \C\}$. Let $f$ be an absolutely continuous $2$-concave increasing function. Let $x\in \F^k$, and $E= E^{(\F)}_k$.
    For all $0 < \delta < 1$, if $\|x\|_{\infty} \leq \delta \|x\|_{\ell_2}$ then 
    \[
    \left[\frac{Ex}{\|x\|_2}\right]_f^f  \leq 
    [g]_f^f + C \cdot \left(-\int_0^{C\delta}  \min(0, f(u)) + \delta\cdot \left(\max(0, f(2\sqrt{\log(1/\delta)})) +2 f'(1)\right)\right),
    \]
     where $g\sim \F \N(0, 1)$ and $C > 0$ is a universal constant.
    In particular if $f = f_p$ for some $-1<p<2$, we have

    \[[Ex]_{f_p}\leq \|x\|_2\cdot (\gamma_{\F,p} + \eps_\delta),\]
    where $\eps_\delta\to 0$ as $\delta\to 0$.
    
\end{lemma}

We prove \cref{lem:A2} in \cref{sec:appendix-A2}. The proof requires a Berry-Esseen type result for test functions of the form $f(\|.\|_2)$  applied to a sum of independent random vectors, which we prove in \cref{sec:appendix-berry-esseen}.

With these results in hand, we can now prove \cref{thm:main-thm}.

\begin{proof}[Proof of \cref{thm:main-thm}]
    We pick $\delta$ to be such that \cref{lem:A2} implies $\|Ex\|_q\leq \|x\|_2\cdot (\gamma_{\F,q} + \eps/2)$ for all $x$ satisfying $\|x\|_\infty\leq \delta\|x\|_{\ell_2}$. 

    We apply \cref{thm:reduction} to the increasing $2$-concave function $f = f_p$, gadget family $\{E^{(\F)}_k\}$, and parameters $\delta$, $\gamma=\gamma_{\F,q} + \eps/2$, and $\eps/2$. By \cref{lem:A2}, the three conditions are satisfied, implying that it is NP-Hard to distinguish the case that $\|A\|_{2\to q}=1$ and $\|A\|_{2\to q}\leq \gamma_{\F,q} + \eps$.
\end{proof}

\subsection{Proof of \cref{thm:reduction}}\label{sec:reduction}
    We will closely follow the arguments used in \cite{briet2015tight, bhattiprolu2018inapproximability}. The starting point of our reduction will be the following result implicit in \cite{briet2015tight}, which informally says that it is NP-hard to find a sparse vector in a subspace, according to a certain block-wise notion of sparsity.

\begin{theorem}[\cite{briet2015tight}]\label{thm:block-l2-l4}
        For all $\eps,\delta, \alpha>0$ and $\F\in \{\R, \C\}$, there is a $k=\poly(1/\eps, 1/\delta, 1/\alpha)$ such that given a subspace $W\subseteq \F^{n\times k}$ in the form of a projection matrix $P\in \F^{(n\times k)\times (n\times k)}$, it is $\PNP$-Hard to distinguish between the following:

        \begin{itemize}
            \item There is a vector $x\in W$ such that for all $i\in [n]$, the vector $x_i\in \F^k$ is in $\{e_1,\ldots, e_k\}$.
            \item For all vectors $x\in W$ with $\|x\|_{\ell_2}^2 = n$, the set
            \[S = \{i\in [n] : \|x_i\|_{\ell_2}\leq 1/\alpha, \|x_i\|_\infty\geq \delta\}\]
            has size at most $\eps n$.
        \end{itemize}
\end{theorem}

    Let $0<\eps', \alpha\leq 1$ be constant parameters depending on $\delta$, $\gamma$, and $\eps$ that we will specify later. We prove hardness of the $[A]_{2\to f}$ problem by a reduction from the NP-Hard problem described in \cref{thm:block-l2-l4} with parameters $\eps', \delta$, and $\alpha$. \cref{thm:block-l2-l4} implies that there is some $k=\poly(1/\eps', 1/\delta, 1/\alpha)$ such that given a projection matrix $P\in \F^{(n\times k)\times (n\times k)}$ for a subspace $W\subseteq \F^{n\times k}$, it is NP-Hard to distinguish between the following:
        \begin{itemize}
            \item There is a vector $x\in W$ such that for all $i\in [n]$, the vector $x_i\in \F^k$ is in $\{e_1,\ldots, e_k\}$.
            \item For all vectors $x\in W$ with $\|x\|_{\ell_2}^2 = n$, the set
            \[S = \{i\in [n] : \|x_i\|_{\ell_2}\leq 1/\alpha, \|x_i\|_\infty \geq \delta\}\]
            has size at most $\eps n$.
        \end{itemize}
        
        Our reduction will map the projection matrix $P\in \F^{(n\times k)\times (n\times k)}$ to the matrix $A=(I_n\otimes E_k)\cdot P$. Note that $A\in \F^{(n\times d_k)\times (n\times k)}$. To analyze the reduction, we must prove completeness and soundness.

    \subsubsection{Completeness}
    
    If there is a vector $x\in W$ such that for all $i\in [n]$, $x_i\in \{e_1,\ldots, e_k\}$, we need to show $[A]_{2\to f}= 1$. 

    Indeed, we can consider the vector $z:=Ax = (E_k\otimes I_n)\cdot Px = (E_k\otimes I_n)x$. We have for all $i\in [n]$, $z_i = E_kx_i$. Since $x_i$ is a standard basis vector, $z_i$ must be equal to some column of $E_k$. So by Assumption \ref{cond:boolean}, all entries of $z$ have magnitude $1/\sqrt{k}$, implying $[z]_f = f^{-1}(f(\frac{1}{\sqrt{k}})) = \frac{1}{\sqrt{k}}$. Therefore $[A]_{2\to f}\geq \frac{[z]_f}{\|x\|_2} = 1$.

    On the other hand, $[A]_{2\to f}\leq \|A\|_{2\to 2} \leq \|P\|_{2\to 2}\cdot \|E_k\|_{2\to 2}\leq 1$ by \cref{claim:2-concave}.

    \subsubsection{Soundness}
    Assuming for all $x\in W$ with $\|x\|_{\ell_2}^2 \leq n$, the set
            \[S = \{i\in [n] : \|x_i\|_{\ell_2} \leq 1/\alpha, \|x_i\|_\infty\geq \delta\}\]
            has size at most $\eps' n$, we need to show $[A]_{2\to f}\leq \gamma + \eps$.

        Let $y\in \F^{n\times k}$ be an arbitrary vector with $\|y\|_2 = 1$, and set $x = \frac{1}{k}\cdot Py$ and $z = Ay = k\cdot (E_k\otimes I_n)x$. Note that because $P$ is a projection matrix, $\|x\|_{\ell_2}^2 = nk\cdot \|x\|_{2}^2 \leq n\|y\|_2^2 = n$, and $\|z\|_2\leq \|y\|_2 = 1$. By virtue of the normalization on $x$, we have $\|x_i\|_{\ell_2} = \|z_i\|_2$ for each block $i\in [n]$.
        
        We must show $[z]_f\leq \gamma + \eps$. We will upper bound the contribution of different indices $i\in [n]$ to $[z]_f$ separately. To do this, define the following partition of $[n]$:

        \[V_0 := S = \{i\in [n] : \|x_i\|_{\ell_2} \leq 1/\alpha, \|x_i\|_\infty\geq \delta\},\]
        \[V_1 := \{i\in [n] : \|x_i\|_{\ell_2} \leq \alpha, \|x_i\|_\infty\leq \delta\alpha\},\]
        \[V_2 := \{i\in [n] : \|x_i\|_{\ell_2} \geq \alpha, \|x_i\|_\infty\leq \delta\alpha\},\]
        \[V_3 := \{i\in [n] : \|x_i\|_{\ell_2} > 1/\alpha\}.\]

        For all $u\in \{0,1,2,3\}$, define $z^{(u)}\in \F^{V_u\times d_k}$ as the collection of $z_i\in \F^{d_k}$ for all $i\in V_u$. Note that 
        \begin{equation}\label{eq:terms}
            [z]_f^f = \sum_{u\in \{0,1,2,3\}} \frac{|V_u|}{|V|} \cdot [z^{(u)}]_f^f.
        \end{equation}

        % We bound each of the four terms in the above sum separately. 
        We will prove bounds on $[z^{(u)}]_f^f$ for $u \in \{0, 1, 2, 3\}$.

        For $u=0$, \cref{claim:2-concave} applied to $z_i$ implies
        \begin{equation}
            [z^{(0)}]_f^f = \E_{i\sim V_0}[z_i]_f^f\leq \E_{i\sim V_0}f(\|z_i\|_2) = \E_{i\sim V_0}f(\|x_i\|_{\ell_2})\leq f(1/\alpha).\label{eq:z0}
        \end{equation}
        
        For $u=1$, a similar application of \cref{claim:2-concave} implies
        \begin{equation}
            [z^{(1)}]_f^f = \E_{i\sim V_1}[z_i]_f^f\leq \E_{i\sim V_1}f(\|z_i\|_2) = \E_{i\sim V_1}f(\|x_i\|_{\ell_2})\leq f(\alpha).\label{eq:z1}
        \end{equation}

        For $u=2$, we have
        \begin{equation}
            [z^{(2)}]_f^f = \E_{i\sim V_2}[z_i]_f^f
            \underset{\text{Assumption \ref{cond:sparse}}}{\leq} \E_{i\sim V_2}\left[f(\gamma\cdot \|z_i\|_2)\right]
            \underset{\cref{claim:2-concave}}{\leq} f(\gamma\cdot \|z^{(2)}\|_2)
            \underset{\|z\|_2 \leq 1}{\leq} f\left(\gamma\cdot \sqrt{\frac{|V|}{|V_2|}}\right).\label{eq:z2}
        \end{equation}
        % \begin{align*}
        %     [z^{(2)}]_f^f &= \E_{i\sim V_2}[z_i]_f^f\\
        %     &\leq \E_{i\sim V_2}\left[f(\gamma\cdot \|z_i\|_2)\right]\tag{Assumption \ref{cond:sparse}}\\
        %     &\leq f(\gamma\cdot \|z^{(2)}\|_2) \tag{\cref{claim:jensen-sq}}\\
        %     &\leq f\left(\gamma\cdot \sqrt{\frac{|V|}{|V_2|}}\right)\tag{$\|z\|_2 \leq 1$}.
        % \end{align*}
        Finally, for $u=3$,
        \begin{align}
            [z^{(3)}]_f^f &= \E_{i\sim V_3}[z_i]_f^f\nonumber\\
            &\underset{\cref{claim:2-concave}}{\leq} \E_{i\sim V_3}f(\|z_i\|_2)\nonumber\\
            &= \E_{i\sim V_3}\left[\|z_i\|_2^2\cdot \frac{f(\|z_i\|_2)}{\|z_i\|_2^2}\right]\nonumber\\
            &\leq \sup_{w\geq 1/\alpha}\frac{f(w)}{w^2}\cdot \E_{i\sim V_3}\|z_i\|_2^2\nonumber\\
            &= \sup_{w\geq 1/\alpha}\frac{f(w)}{w^2}\cdot \|z^{(3)}\|_2^2\nonumber\\
            &\underset{\|z\|_2^2 = 1}{\leq} \sup_{w\geq 1/\alpha}\frac{f(w)}{w^2}\cdot \frac{|V|}{|V_3|}.\label{eq:z3}
        \end{align}

Now we are equipped to bound \cref{eq:terms}.

\begin{align*}
    [z]_f^f &\leq \sum_{u\in \{0,1,2\}}\frac{|V_u|}{|V\setminus V_3|}[z^{(u)}]_f^f + \sup_{w\geq 1/\alpha}\frac{f(w)}{w^2}\tag{\cref{eq:terms,eq:z3}}\\
    &\leq f\left(\sqrt{\sum_{u\in \{0,1,2\}}\frac{|V_u|}{|V\setminus V_3|}[z^{(u)}]_f^2}\right)+\sup_{w\geq 1/\alpha}\frac{f(w)}{w^2}\tag{Jensen's inequality for $x\to f(\sqrt{x})$}\\
    &\leq f\left(\sqrt{\sum_{u\in \{0,1,2\}}\frac{|V_u|}{(1-\alpha^2)|V|}[z^{(u)}]_f^2}\right)+\sup_{w\geq 1/\alpha}\frac{f(w)}{w^2}\tag{$|V_3|\leq \alpha^2\cdot |V_0|$ by def of $V_3$}\\
    &\leq f\left(\sqrt{ \frac{\eps'/\alpha^2 + \alpha^2 + \gamma^2}{(1-\alpha^2)}}\right) + \sup_{w\geq 1/\alpha}\frac{f(w)}{w^2}\tag{\cref{eq:z0,eq:z1,eq:z2} and $V_0\leq \eps'|V|$}\\
    &\leq f\left(\frac{\gamma + \sqrt{\eps'}/\alpha + \alpha}{\sqrt{1-\alpha^2}}\right)+ \sup_{w\geq 1/\alpha}\frac{f(w)}{w^2}\tag{$\sqrt{a+b}\leq \sqrt{a} + \sqrt{b}$, $f$ is monotone}\\
    &= f\left(\frac{\gamma + 2\alpha}{\sqrt{1-\alpha^2}}\right) + \sup_{w\geq 1/\alpha}\frac{f(w)}{w^2}.\tag{Setting $\eps' = \alpha^4$}
\end{align*}
%In fourth inequality we use  and  the soundness assumption $|V_0| \leq \eps'  |V|$ and 
%            $|V_3| \leq  \alpha^2 |V|$ (by definition of $V_3$) which together imply
 %           $\frac{|V_0|}{|V \setminus V_3|} \leq \frac{\epsilon'}{1 - \alpha^2}.$

Using the assumption that $f$ is continuous and $\lim_{x\to \infty}f(x)/x^2 = 0$, we get that the limit of the right hand side as $\alpha\to 0$ is exactly $f(\gamma)$. Therefore, there exists some $\alpha>0$ independent of $n$ such that $[z]_f\leq \gamma + \eps$. We choose $\alpha$ in our invocation of \cref{thm:block-l2-l4} accordingly, completing the proof that $[A]_{2\to f}\leq \gamma + \eps$.

\subsection{Proof of \cref{thm:per-thm}}
    \label{sec:permanent-proofs}
    
 In this section we prove \cref{thm:per-thm}.  We use the following lemma which proves that the approximability of the permanent of highly rank-deficient $n\times n$ PSD matrices is essentially the same as the approximability of of the $2\to 0$ norm.
    
    \begin{lemma}\label{lem:permanent-equivalence}
        Let $V\in \C^{n\times d}$. Then,
        \[c_{n,d}^{\C}\cdot \binom{n+d-1}{d}^{-1}\cdot \|V\|_{2\to 0}^{2n}\leq \per(VV^\dag)\leq c_{n,d}^{\C}\cdot \|V\|_{2\to 0}^{2n},\]
        where $c_{n,d}^{\C}$ is defined in \cref{lem:cnd}.
    \end{lemma}

    \begin{proof}
        Let $v_1,\ldots, v_n$ be the rows of $V$. For the upper bound, we can write 
        \begin{align*}
            \per(V V^\dag) &= \E_{x\sim \C N(0,I)} \prod_{i\in [n]}|\langle x, v_i\rangle|^2\\
            &= c_{n,d}\cdot \E_{\|x\|_2=1} \prod_{i\in [n]}|\langle x, v_i\rangle|^2\tag{\cref{lem:cnd}}\\
            &\leq c_{n,d}\cdot \max_{\|x\|_2 = 1} \prod_{i\in [n]}|\langle x, v_i\rangle|^2\\
            &= c_{n,d}\cdot \|V\|_{2\to 0}^{2n}.
        \end{align*}

        Next we prove the lower bound. Let $z\in \C^d$ be a vector with $\|z\|_2 = 1$ vector maximizing $\|Vz\|_{0}=\prod_{i\in [n]}|\langle z, v_i\rangle|^{1/n}$.
        We have $zz^\dag \preceq \|z\|_{\ell_2}^2\cdot I = d\cdot I$, so $V z z^\dag V^\dag\preceq d\cdot V V^\dag$. By \cref{lem:permanent-monotone}, this implies $\per(V z z^\dag V^\dag)\leq d^n\cdot \per(VV^\dag)$. Since $V z z^\dag V^\dag$ is rank 1, we can compute its permanent as $\per(V z z^\dag V^\dag) = n!\cdot \prod_{i\in [n]}|\langle z, v_i\rangle|^2$.

        \begin{align*}
            \|Vz\|_{0}^{2n}&=\max_{\|x\|_2=1}\prod_{i\in [n]}|\langle x, v_i\rangle|^2\\
            &= \frac{1}{n!}\cdot \per(V zz^\dag V^\dag)\\
            &\leq \frac{d^n}{n!}\cdot \per(V V^\dag)\\
            &= \binom{n+d-1}{d}\cdot c_{n,d}^{-1}\cdot \per(V V^\dag).
        \end{align*}
    \end{proof}

    \begin{proof}[Proof of \cref{thm:per-thm}]
        We start from \cref{thm:main-thm} for the case $\F = \C$ and $q=0$, to get that it is NP-hard to approximate $\|A\|_{2\to 0}$ within a factor of $e^{-\gamma/2 + \eps/4}$ (recall that $\gamma_{\C, 0} = e^{-\gamma/2}$ by \cref{fact:gamma-values}). 

        We reduce the problem of approximating $\|A\|_{2\to 0}$ for $A\in \C^{n\times d}$ to approximating the permanent of the positive semidefinite matrix $B=A^{(k)}(A^{(k)})^\dag$, where $A^{(k)}\in \C^{nk\times d}$ is as in \cref{claim:tensor}. By \cref{lem:permanent-equivalence}, $\per(B)$ is proportional to $\|A^{(k)}\|_{2\to 0}^{2nk}$ up to a multiplicative error of 
        
        \[\binom{nk+d-1}{d}\leq e^{\eps kn/2},\]

        which holds for $k = O(\frac{d}{n\eps^2})$ and $\eps>0$ small enough. Note that the reduction is efficient because $k$ is polynomial in the size of $A$.
        
        By \cref{claim:tensor}, we have $\|A^{(k)}\|_{2\to 0}^{2nk}= \|A\|_{2\to 0}^{2nk}$ which is hard to approximate within a factor of $e^{-kn(\gamma-\eps/2)}$.
        Therefore, it is NP-hard to approximate $\per(B)$ within a factor of $e^{-kn(\gamma - \eps)}$, where $B\in \C^{kn\times kn}$.
    \end{proof}

\printbibliography
\appendix

\section{Proof of \cref{lem:A2}}\label{sec:appendix-A2}

In this section we prove \cref{lem:A2}. We use the following corollary which we will prove in \cref{sec:appendix-berry-esseen}.

\begin{corollary}
\label{cor:ub}
    Let  $f:\R_{> 0}\to \R$ be an absolutely continuous $2$-concave increasing function with $0\in \overline{\range(f)}$. Let $\F \in \{\R, \C\}$, and let $z\in \F^n$ be a vector with $\|z\|_{\ell_2}= 1$ and $\|z\|_\infty\leq \delta$. For each $i$, let $\sigma_i$ be an independently and uniformly sampled member of $\{-1, +1\}$ if $\F = \R$ and be an independently and uniformly sampled member of $\{-1, +1, -i, +i\}$ otherwise. Then there exists a universal $C > 0$ such that $Z = \sum_{i \in [n]} \sigma_i z_i$ satisfies
    \[[Z]_f^f\leq [g]_f^f + C\cdot \left(- \int_0^{C\delta}  \min(0, f(u)) + \delta\cdot \left(f(2\sqrt{\log(1/\delta)}) + f'(1)\right)\right),\]
    where $g\sim \F \N(0, 1)$. %when $\F = \R$ and $g\sim \C \N(0, 1)$ when $\F = \C$.
\end{corollary}

Now we are ready to prove the main result of this section.
\begin{proof}[Proof of \cref{lem:A2}]
    % Fix some $0<p\leq 1$. 
    % Let us assume $\|x\|_\infty \leq \delta \|x\|_2$ for $\delta$ to be fixed later.
    % For $j \in [n]$ let us define $z_j \coloneqq \frac{x_j}{\|x\|_2}$.
    Let $z = x / \|x\|_2$ and $y = Ez = Ex/ \|x\|_2$. Let $m$ be the number of rows of $E$. Note that for  $i \in [m]$, we have
    \[
    y_i = \sum_{j \in [n]} E_{i,j}\cdot z_j.% = \frac{(Ex)_i}{\|x\|_2}.
    \]
    % For $a, b: [m] \to \C$ define
    % \[
    % \langle a, b \rangle = \E_{i \sim A} a_i \bar{b_i}.
    % \]
    Let $Z_j = E_{i,j} \cdot z_j$ be the random variable where $i$ chosen uniformly at random from $[m]$. % Let $\beta_j = |Z_j|$. 
    We invoke \cref{cor:ub} on $z/\sqrt{k}$. We verify its conditions:
    First,  $\norm{\frac{z}{\sqrt{k}}}_{\ell_2} = \norm{z}_2=1$. Second, by definition of $E_k^{(\F)}$ we can write $Z_j =  \sigma_j \cdot \frac{z_j}{\sqrt{k}}$, where the random variables $\sigma_j\in \{+1,-1\}$ when $\F=\R$ and $\sigma_j\in \{+1,-1,+i,-i\}$ chosen  independently and uniformly at random 
    Lastly,  
    $$\frac{|z_j|}{\sqrt{k}} = \frac{|x_j|}{\sqrt{k}\norm{x}_2} = \frac{|x_j|}{\norm{x}_{\ell_2}} \leq \frac{\norm{x}_\infty}{\norm{x}_{\ell_2}}\leq \delta.$$
    % Since $E \in \{E_$
    %Note that as $|E_{ij}| = \frac{1}{\sqrt{n}}$ we have $|Z_j| = \frac{|z_j|}{\sqrt{n}}$ almost surely. 
    % Second, 
    % \[
    % \sum_{j\in [n]} |Z_j|^2 = \frac{1}{n} \sum_{j \in [n]}|z_j|^2 = 1. 
    % \]
    % Furthermore,
    % \begin{align*}
    %         \sum_{j \in [n]} \E[|Z_j|^2] = \sum_{j \in [n]} \E[Z_j \bar{Z_j}] =  
    %         \frac{1}{m}\sum_{j\in [n]} \sum_{i\in [m]} E_{ij}z_j \overline{E_{ij}z_j}
    %         = \sum_{j\in [n]}  z_j\langle E^j,E^j\rangle \bar{z}_j
    %         =\frac1{n}\sum_j |z_j|^2 = 1,
    %         %\E_{i \sim [m]} \sum_{j \in [n]} Z_i \bar{Z_i} = \E_{i \sim [m]} \sum_{j \in [n]} E_{ij} z_j \overline{E_{ij} z_j} %\sum_{i \in B} \frac{|\widehat{f}(i)|^2}{\|\widehat{f}\|_{\ell_2}^2} = 1
    % \end{align*}
    % where $E^j$ denotes the $j$th column of $E$ and in the penultimate equality we used that $E^\dagger E = \frac{m}{n} I$.
    % Moreover,  we have
    % \[
    % \sum_{j \in [n]} |Z_j|^3 \leq (\sum_{j \in [n]}  |Z_j|^2) \cdot \max_{j \in [n]} |Z_j| \leq \delta.
    % \]
    % where the penultimate inequality follows from the Cauchy-Schwarz inequality.
    Now by invoking \cref{lem:A1} on the random variables $Z_1, \ldots, Z_k$, we get
    \begin{align*}
    \left[\frac{Ex}{\|x\|_2}\right]_f^f = [y]_f^f = \E_{i \sim [m]} f(y_i) = \E \big[ f\big(\sum_{j \in [n]} Z_j\big) \big]   \leq 
    [g]_f^f + C\cdot \left(- \int_0^{C\delta}  \min(0, f(u)) + \delta\cdot \left(f(2\sqrt{\log(1/\delta)}) + f'(1)\right)\right),
    \end{align*}
    as desired.

    Now assume $f = f_p$ for some $-1<p<2$. By \cref{lem:homogeneous}, $\|\cdot\|_p$ is homogeneous, and we get
    \begin{align*}
        [Ex]_f &= \|x\|_2\cdot \left[\frac{Ex}{\|x\|_2}\right]_f\\
        &\leq \|x\|_2\cdot f^{-1}\left(\gamma_{\F, p}^p +  C\cdot \left(- \int_0^{C\delta}  \min(0, f(u)) + \delta\cdot \left(f(2\sqrt{\log(1/\delta)}) + f'(1)\right)\right)\right)\\
        % &=: \|x\|_2\cdot f^{-1}\left([g^{\F}]_f^f + \eps_{\delta}\right)
    \end{align*}

    The second term is $0$ if $p>0$, otherwise it is equal to $\frac{\delta^{p+1}}{p+1}$. The third term is $2\delta$ if $p<0$, otherwise it is bounded by $4\delta(\log(1/\delta)+1)$ for $\delta$ small enough. Therefore, the limit of the right hand side as $\delta\to 0$ is $\|x\|_2\cdot \gamma_{F,p}$.
\end{proof}

\section{Proof of \cref{cor:ub}}\label{sec:appendix-berry-esseen}

\cref{cor:ub} follows directly as a special case of the following Lemma, for $k=1$ if $\F=\R$ and for $k=2$ if $\F=\C$.

\begin{restatable}{lemma}{main}
\label{lem:A1}
    Let  $f:\R_{> 0}\to \R$ be an absolutely continuous $2$-concave increasing function with $0\in \overline{\range(f)}$. Let $k \geq 1$ and let $X_1, \ldots, X_n$ are bounded independent random variables in $\R^k$ with $ \E[X_i] = 0$, $\text{Cov}(\sum_i X_i) = I_k/k$, and $\|X_i\|_{\ell_2} \leq \delta_i$ such that $\sum_i \delta_i^2\leq 1$ and $\delta_i\leq \delta$ for some $0 < \delta < 1$, then there exists $\eta_k > 0$ such that
    %\[\left[\sum X_i\right]_f^f\leq 
    \[
%\E \big[f(\|\sum_{i \in [n]} X_i\|_{\ell_2})\big] 
\left[\big\|\sum_{i \in [n]} X_i\big\|_{\ell_2}\right]_f^f  
    \leq 
    \left[\|g\|_{\ell_2}\right]_f^f - e\int_0^{C_k\delta/e}  \min(0, f(u))du + \delta\cdot \left(C_k\cdot \max(0, f(2\sqrt{\log(1/\delta)})) +2 f'(1)\right),\]

    where $g\sim N(0, I_k/k)$, and $C_k$ is the constant in \cref{thm:berry-esseen}.
\end{restatable}

Before proving \cref{lem:A1}, we state some probabilistic tools we will require in the proof.

\begin{theorem}[Multivariate Berry-Esseen \cite{bentkus2005lyapunov}]
\label{thm:berry-esseen}
    Let $k \geq 1$ and let $X_1, \ldots, X_n$ be independent random variables in $\R^k$ satisfying $\E[X_i] = 0$ for $1 \leq i \leq n$. Define $X = X_1 + \cdots + X_n$ and suppose  $\E[XX^T] = I_k$. Further let $g \sim \N(0, I_k)$. Then there exists $C_k > 0$ such that for all convex sets $U \subseteq \R^k$ it holds that  
    \[
    |\Pr[g \in U] - \Pr[X \in U]| \leq C_k \cdot \E_{i \in [n]}[\|X_i\|_{\ell_2}^3].
    \]
    In particular,
     for all $u \geq 0$ it holds that  
    % \[
    % |\Pr[g \in U] - \Pr[X \in U]| \leq C_k \cdot \E_{i \in [n]}[(\sqrt{k} \|X_i\|_2)^3].
    % \]
    \[
    |\Pr[\|g\|_{\ell_2} \leq u] - \Pr[\|X\|_{\ell_2} \leq u]| \leq C_k \cdot \E_{i \in [n]}[ \|X_i\|_{\ell_2}^3].
    \]
    % TODO: normalization
\end{theorem}

\begin{theorem}[{Multivariate Hoeffding \cite[Theorem 3]{Pin}}]\label{thm:mult-hoeff}
    Let $k \geq 1$ and let $X_1, \ldots, X_n$ be independent random variables in $\R^k$ satisfying $\E[X_i] = 0$ for $1 \leq i \leq n$. Further suppose $\|X_i\|_{\ell_2} \leq \delta_i$ almost surely for $\delta_1, \ldots, \delta_n > 0$. Define $X = X_1 + \cdots + X_n$. % and suppose  $\E[XX^T] = I_k$. 
    % Further let $g \sim \N(0, I_k)$. Then there exists $C_k > 0$ such that for all convex sets $U \subseteq \R^k$ it holds that  
    \[
    |\Pr[\|X\|_{\ell_2} \geq u]| \leq 2 \cdot e^{-u^2/2 \sum_{i \in [n]} \delta_i^2}.
    \] % https://link.springer.com/chapter/10.1007/978-1-4612-0367-4_9
\end{theorem}

\begin{fact}\label{fact:limufu}
    Let $f:\R_{>0}\to \R$ be an absolutely continuous function. If $\int_0^b f(u)$ is bounded then $\lim_{u\to 0} uf(u)=0$.
\end{fact}

\begin{claim}
\label{claim:expectation-integral-equality}
    Let $X$ be a random variable over $\R_{> 0}$, and let $f:{\R_{> 0}} \to \R$ be an absolutely continuous function with $f(a)=0$ for some $a\in \R_{>0}$. Then, if $\E[f(X)]$ exists,
    \[
    \E[f(X)] = -\int_0^a f'(u)\Pr[X\leq u]du+ \int_a^{\infty}f'(u)\Pr[X\geq u]du+ \lim_{u\to 0}f(u)\Pr[X\leq u] - \lim_{u\to \infty}f(u)\Pr[X\geq u].
    \]
\end{claim}
\begin{proof}
Let $\mu:\R_{>0}\to \R$ denote the PDF of the random variable $X$.
    \begin{align*}
    % \left[\sum X_i\right]_f^f
    \E[f(X)] &= \int_0^\infty f(u) d\mu(u)\\
    &=\int_{0}^af(u)d\mu(u) + \int_a^\infty f(u)d\mu(u)\\
    &= f(u)\Pr[X\leq u]\Big|_0^a - \int_{I_1}f'(u)\Pr[X\leq u] - f(u)\Pr[X\geq u]\Big|_a^\infty + \int_{I_2}f'(u)\Pr[X\geq u]\\
    &= \lim_{u\to 0}f(u)\Pr[X\leq u] - \lim_{u\to \infty}f(u)\Pr[X\geq u]- \int_{I_1}f'(u)\Pr[X\leq u]+ \int_{I_2}f'(u)\Pr[X\geq u]\tag{$f(a)=0$.}\\
\end{align*}
\end{proof}
% \begin{fact}[\cite{dasgupta2003elementary}]
%     Let $t \sim \chi^2(k)$, where $\chi^2(k)$ is the chi-squared distribution with $k$ degrees of 
% \end{fact}
\begin{fact}\label{fact:chi2-bound}
    Let $k \geq 1$ and let $g\sim N(0, I_k/k)$. Then for all $0\leq u\leq 1$, $\Pr[\|g\|_{\ell_2}\leq u]\leq eu$.
\end{fact}
\begin{proof}
    $k\cdot \|g\|_{\ell_2}^2$ is distributed as a Chi-squared random variable with $k$ degrees of freedom. In \cite[Lemma~2.2]{dasgupta2003elementary}, the authors show that
    \[
    \Pr[\|g\|_{\ell_2}\leq u] \leq (u^2 e^{1-u^2})^{k/2}\leq e\cdot u,
    \]
    as desired.
    % So, the p.d.f. of $\|g\|_{\ell_2}^2$ is
    % \[p(x) = \frac{k}{2^{k/2}\Gamma(k/2)}\cdot \left(kx\right)^{k/2-1}\cdot e^{-kx}\]
\end{proof}

With these facts in hand, we are ready to prove \cref{lem:A1}.

\begin{proof}[Proof of \cref{lem:A1}]

Define the random variable $X =  \|\sum_i X_i\|_{\ell_2}$. 
We split the domain of $f$ into $I_1 = f^{-1}([-\infty, 0])$ and $I_2 = f^{-1}([0,\infty])$, where $I_1 = [0,a]$ and $I_2 = [a,\infty]$. Since $0\in \overline{\range(f)}$, we have $f(a) = 0$. We use \cref{claim:expectation-integral-equality} to write the left hand side as
\begin{align*}
    % \left[\sum X_i\right]_f^f
    \E[f(X)] %&= \E\left[f\left(X\right)\right]\\
    % &= \int_{I_1}f(u)\Pr\left[X = u\right]du + \int_{I_2}f(u)\Pr\left[X = u\right]du\\
    % &= f(u)\Pr[X\leq u]\Big|_0^a - \int_{I_1}f'(u)\Pr[X\leq u] - f(u)\Pr[X\geq u]\Big|_a^\infty + \int_{I_2}f'(u)\Pr[X\geq u]\\
    &= - \int_{0}^a f'(u)\Pr[X\leq u]+ \int_a^\infty f'(u)\Pr[X\geq u] + \lim_{u\to 0}f(u)\Pr[X\leq u] - \lim_{u\to \infty}f(u)\Pr[X\geq u]\\
    &\leq - \int_0^a f'(u)\Pr[X\leq u]+ \int_a^\infty f'(u)\Pr[X\geq u].\tag{$f(0)\leq f(a) = 0$, $f(\infty)\geq f(a) = 0$.}
\end{align*}

We bound each of the above integrals separately.

\begin{claim}\label{claim:int1}
We have
    \[\int_{0}^a f'(u)\Pr[X\leq u] \geq  \int_0^a f'(u)\Pr[\|g\|_{\ell_2}\leq u] + e\int_0^{C_k\delta/e}  \min(0, f(u)).\]
\end{claim}

\begin{claim}\label{claim:int2}
We have
    \[
 \int_{a}^\infty f'(u)\Pr[X\geq u] \leq \int_a^\infty f'(u)\Pr[|g|\geq u] + \delta\cdot \left(C_k\cdot \max(0, f(2\sqrt{\log(1/\delta)})) +2 f'(1)\right).
    \]
\end{claim}

We will complete the proof of \cref{lem:A1} using these two claims, and prove the claims later.

\begin{align*}
    \E[f(X)] &\leq - \int_0^a f'(u)\Pr[X\leq u]+ \int_a^\infty f'(u)\Pr[X\geq u]\\
    &\leq -\int_0^a f'(u)\Pr[\|g\|_{\ell_2}\leq u] - e\int_0^{C_k\delta/e}  \min(0, f(u)) \\
    &\qquad+\int_a^\infty f'(u)\Pr[|g|\geq u] + \delta\cdot \left(C_k\cdot f(2\sqrt{\log(1/\delta)}) +2 f'(1)\right)\\
    &=\E[f(\|g\|_{\ell_2})] - \int_0^{ C_k\delta}  \min(0, f(u)) + \delta\cdot \left(C_k\cdot \max(0, f(2\sqrt{\log(1/\delta)})) +2 f'(1)\right).
\end{align*}

Here, the last equality is by applying \cref{claim:expectation-integral-equality} on the random variable $\|g\|_{\ell_2}$:
\begin{align*}
    \E[\|g\|_{\ell_2}]= &-\int_0^a f'(u)\Pr[\|g\|_{\ell_2}\leq u] + \int_a^\infty f'(u)\Pr[|g|\geq u]\\
     &+ \lim_{u\to 0}f(u)\Pr[\|g\|_{\ell_2}\leq u] - \lim_{u\to \infty}f(u)\Pr[\|g\|_{\ell_2}\geq u]
\end{align*}

Observe that by \cref{fact:limufu}, the third term above is zero, and by \cref{claim:2-concave}, \[\lim_{u\to \infty}f(u)\Pr[\|g\|_{\ell_2}\geq u]\leq \lim_{u\to \infty}(f'(1)\cdot u^2+f(1))\cdot \Pr[\|g\|_{\ell_2}\geq u] = 0.\]

This completes the justification of the last inequality. It remains to prove \cref{claim:int1} and \cref{claim:int2}. We begin by writing an inequality which will be used in both proofs. By \cref{thm:berry-esseen}, for any $u\geq 0$,

\begin{align}
    |\Pr[X\leq u]-\Pr[\|g\|_{\ell_2}\leq u]|&\leq C_k\cdot \E_{i \in [n]} [\|X_i\|_{\ell_2}^3]\nonumber\\
    &\leq C_k\cdot \sum_{i \in [n]} [\|X_i\|_{\ell_2}^2] \cdot \max_{i \in [n]}\|X_i\|_{\ell_2}\nonumber \\
    &\leq C_k\cdot \delta.\label{eq:berry-esseen-error}
\end{align}

\begin{proof}[Proof of \cref{claim:int1}]
    For a parameter $b$ with $0\leq b\leq a$, we write
\begin{align*}
    \int_{I_1}f'(u)\Pr[X\leq u]&\geq \int_b^a f'(u)\Pr[X\leq u]\tag{$f'(u)\geq 0$}\\
    &\geq \int_b^a f'(u) (\Pr[\|g\|_{\ell_2}\leq u] - C_k \delta)\tag{\cref{eq:berry-esseen-error}}\\
    &=\int_0^a f'(u)\Pr[\|g\|_{\ell_2}\leq u]-\int_0^b f'(u)\Pr[\|g\|_{\ell_2}\leq u] - C_k\delta (f(a)-f(b)) \\
    &\geq \int_0^a f'(u)\Pr[\|g\|_{\ell_2}\leq u] - e\int_0^b f'(u) u +C_k\delta f(b) \tag{using $f(a)=0$, and \cref{fact:chi2-bound}}\\
    &= \int_0^a f'(u)\Pr[\|g\|_{\ell_2}\leq u] - euf(u) \Big|_0^b + e\int_0^b f(u) + C_k  \delta f(b) \\
    &= \int_0^a f'(u)\Pr[\|g\|_{\ell_2}\leq u] - ebf(b) + e\int_0^b f(u) + C_k  \delta f(b) \tag{\cref{fact:limufu}}\\
    &= \int_0^a f'(u)\Pr[\|g\|_{\ell_2}\leq u] + e\int_0^{\min(a, C_k\delta/e)}  f(u)\tag{Setting $b=\min(a, C_k\delta/e)$}
    \label{eq:I1integral} \\
    &= \int_0^a f'(u)\Pr[\|g\|_{\ell_2}\leq u] + e\int_0^{C_k\delta/e}  \min(f(u), 0).
\end{align*}

The penultimate equality is because if $\min(a, C_k\delta/e)=a$, then $f(b) = 0$, and if $\min(a, C_k\delta/e)=C_k\delta/e$, then $bf(b) = C_k\delta f(b)/e$.

\end{proof}

\begin{proof}[Proof of \cref{claim:int2}]
    Applying \cref{thm:mult-hoeff} on $X$, 
    \begin{equation}\label{eq:hoeff}
        \Pr[X\geq u]\leq 2\cdot e^{-u^2/2\sum \delta_i^2} \leq 2\cdot e^{-u^2/2}.
    \end{equation}

    For $\max(a, 1) <c<\infty$ we write
\begin{align}
    \int_{I_2}f'(u)\Pr[X\geq u] &= \int_{a}^c f'(u)\Pr[X\geq u] + \int_{c}^\infty f'(u)\Pr[X\geq u]\nonumber \\
    &\leq \int_a^c f'(u)(\Pr[|g|\geq u]+C_k\delta) + \int_c^\infty 2f'(u)e^{-u^2/2}\nonumber \tag{\cref{eq:berry-esseen-error}, \cref{eq:hoeff}}\\
    &\leq \int_a^c f'(u)\Pr[|g|\geq u] +  C_k\delta\cdot (f(c)-f(a)) +2 f'(1) \cdot \int_c^\infty  ue^{-u^2/2} \tag{using $f$ is 2-concave, \cref{claim:2-concave}, and $c\geq 1$}\nonumber \\
    &\leq \int_a^\infty f'(u)\Pr[|g|\geq u] + C_k\delta\cdot f(c) +2 f'(1) \cdot \int_c^\infty ue^{-u^2/2}\tag{$f$ is increasing, $f(a)=0$}\nonumber \\
    &= \int_a^\infty f'(u)\Pr[|g|\geq u] + C_k\delta\cdot f(c) +2 f'(1) \cdot e^{-c^2/2}%\\
    %&=  \int_a^c f'(u)\Pr[|g|\geq u]+2\delta f(c) + \frac{2'f(c)}{c} 
\end{align}

Setting $c = \max(a, 2\sqrt{\log(1/\delta)})$, we get
\begin{equation}\label{eq:I2integral}
    \int_{I_2}f'(u)\Pr[X\geq u] \leq \int_a^\infty f'(u)\Pr[|g|\geq u] + \delta\cdot \left(C_k\cdot \max(0, f(2\sqrt{\log(1/\delta)})) +2 f'(1)\right).
\end{equation}

The inequality above is by bounding $e^{-c^2/2}\leq \delta$, and if $\max(a, 2\sqrt{\log(1/\delta)}) = a$, then $f(c)=0$.
\end{proof}

\end{proof}

\section{Algorithm for approximating $2 \to q$ norm}\label{appendix:alg}

Let $q< 2$. Given a matrix $A\in \F^{n\times d}$ with rows $\{a_i\}_{i\in [n]}$, we write an expression for $\|A\|_{2\to q}$:
\[\|A\|_{2\to q} = f_q^{-1}\left(\max_{x\in \F^d, \|x\|_2 = 1}\left[\sum_{i\in [n]}f_q\left(|\langle a_i, x\rangle|\right)\right]\right) = f_q^{-1}\left(\max_{x\in \F^d, \|x\|_2 = 1}\left[\sum_{i\in [n]}f_q\left(\sqrt{a_i^\dag x x^\dag a_i}\right)\right]\right).\]

Notice that $xx^\dag$ is a rank-$1$ PSD matrix with $\tr(xx^\dag) = x^\dag x = \|x\|_{\ell_2} = d\cdot \|x\|_2 = d$. We can relax the rank $1$ constraint to obtain a relaxation of $\|A\|_{2\to q}$ which we denote by $\SDP_{2\to q}(A)$:

\[\SDP_{2\to q}(A) := f_q^{-1}\left(\max_{X\succeq 0, \tr(X) = d}\left[\sum_{i\in [n]}{f_q\left(\sqrt{a_i^\dag X a_i}\right)}\right]\right),\]
where the maximum is taken over all matrices in $\F^{d\times d}$. Note that the objective function being maximized is concave for all $q\leq 2$ because the function $X\to a_i^\dag X a_i$ is linear and $f_q$ is $2$-concave.
\begin{lemma}\label{lem:alg-tight}
    For any matrix $A\in \F^{n\times d}$ and any $-1<q\leq 2$,
    \[\|A\|_{2\to q}\leq \SDP_{2\to q}(A)\leq \gamma_{\F, q}^{-1}\cdot \|A\|_{2\to q}.\]
\end{lemma}

\begin{proof}
    The first inequality is because $\SDP_{2\to q}(A)$ is a relaxation of $\|A\|_{2\to q}$. For the second inequality, we describe a rounding procedure for the relaxation.

    Let $X^*\in \F^{d\times d}$ be the optimal solution to $\SDP_{2\to q}(A)$. We sample a vector $x\sim \F\N(0,X^*)$. Clearly,

    \begin{equation}
        \label{eqn:alg-max}\|A\|_{2\to q}^2 = \max_{x\in \F^n, x\neq 0}\frac{\|Ax\|_q^2}{\|x\|_2^2}\geq \frac{\E \|Ax\|_q^2}{\E \|x\|_2^2}.
    \end{equation}

    First we calculate the denominator of the right hand side: $\E \|x\|_2^2 = \frac{1}{d}\cdot \tr(X) = 1$. For the numerator, we bound

    \begin{align*}
        \E \|Ax\|_q^2 &= \E \left(f_q^{-1}\left(\sum_{i\in [n]}f_q(|\langle a_i, x\rangle|)\right)^2\right)\\
        &\geq f^{-1}_q\left(\sum_{i\in [n]}\E f_q(|\langle a_i, x\rangle |)\right)^2\tag{Jensen's inequality, and $x\to f^{-1}_q(x)^2$ is convex}\\
        &= f^{-1}_q\left(\sum_{i\in [n]}\E_{g\sim \F\N(0, 1)}f_q\left(\sqrt{a_i^\dag X^* a_i}\cdot |g|\right)\right)^2\tag{$\langle a_i, x\rangle$ is a Gaussian with variance $a_i^\dag X^* a_i$}\\
        &= f_q^{-1}\left(\E_{g\sim \F\N(0, 1)} f_q(|g|)\right)^2\cdot f^{-1}_q\left(\sum_{i\in [n]}f_q\left(\sqrt{a_i^\dag X^* a_i}\right)\right)^2\tag{Homogeneity}\\
        &= \gamma_{\F, q}^2\cdot \SDP_{2\to q}(A)^2.
    \end{align*}

    Together with \cref{eqn:alg-max}, this implies $\|A\|_{2\to q}\geq \gamma_{\F, q}\cdot \SDP_{2\to q}(A)$, completing the proof.
    
\end{proof}

\end{document}